%% file: main.tex
\documentclass[12pt]{article}
\usepackage{amsmath}
\usepackage{graphicx,psfrag,epsf}
\usepackage{enumerate}
\usepackage{natbib}
\usepackage{url} 
\usepackage{amsmath,amssymb}
\usepackage{hyperref} 
\usepackage{subcaption}
\usepackage{listings}
\usepackage{xcolor}
\usepackage{array}          
\usepackage{enumitem}
\usepackage{etoolbox}       
\usepackage[ruled,vlined]{algorithm2e} 
\usepackage{amsthm}
\usepackage{comment}
\newcommand{\blind}{0}

\newtheorem{definition}{Definition}[section]
\newtheorem{lemma}[definition]{Lemma} 
\newtheorem{theorem}{Theorem}[section]
 
\newtheorem{proposition}{Proposition}[section]

\usepackage{booktabs, tabularx}
\usepackage{xparse}  

\newcolumntype{L}{>{\raggedright\arraybackslash}X}
\newcolumntype{C}{>{$}l<{$}} 
\newcolumntype{Y}{>{\raggedright\arraybackslash}X}

\NewDocumentEnvironment{showas}{ m m m m O{} }{%
  \begingroup
  \renewcommand{\thesection}{#1}%
  \setcounter{#4}{\numexpr #2-1\relax}%
  \csname begin\endcsname{#3}[#5]%
}{%
  \csname end\endcsname{#3}%
  \endgroup
}

\counterwithin{table}{section}

\begin{document}

\def\spacingset#1{\renewcommand{\baselinestretch}%
{#1}\small\normalsize} \spacingset{1}




\if0\blind
{
  \title{\bf Differentially Private Computation of the Gini Index for Income Inequality}
  \author{Wenjie Lan\\
    Department of Statistical Science, Duke University\\
    and \\
    Jerome P. Reiter \\
    Department of Statistical Science, Duke University }
  \maketitle
} \fi

\if1\blind
{
  \bigskip
  \bigskip
  \bigskip
  \begin{center}
    \title{\bf Differentially Private Computation of the Gini Index for Income Inequality}
\end{center}
  \medskip
\maketitle
} \fi

\bigskip
\begin{abstract}
The Gini index is a widely reported measure of income inequality. In some settings, the underlying data used to compute the Gini index are confidential.
The organization charged with reporting the Gini index may be concerned that its release 
could leak information about the underlying  data. We present an approach for bounding this information leakage by releasing a differentially private version of the Gini index. In doing so, we 
analyze how adding, deleting, or altering a single observation in any specific dataset can affect the computation of the Gini index; this is known as the local sensitivity. We then derive a smooth upper bound on the local sensitivity. Using this bound, we define a mechanism that adds noise to the Gini index, thereby satisfying differential privacy.
Using simulated and genuine income data, we show that the mechanism can reduce the errors from noise injection substantially relative to differentially private algorithms that rely on the global sensitivity, that is, the maximum of the local sensitivities over all possible datasets. 
We characterize settings where using smooth sensitivity 
can provide highly accurate estimates, as well as settings where the noise variance is simply too large to provide reliably useful results. 
We also present a 
Bayesian post-processing step that provides interval estimates about the value of the Gini index computed with the confidential data. 
\end{abstract}

\noindent%
{\it Keywords:}  confidentiality; perturbation; privacy; smooth sensitivity. 
\vfill
\newpage
\spacingset{1.45} 
\section{Introduction}
\label{sec:intro}

The Gini index  is a widely used statistic for quantifying inequality of monetary variables, such as income and wealth, within some defined group, such as a state or country \citep{rothschild1973some}.  Its value ranges from zero to one, with zero representing perfect equality, e.g., everyone has the same income, and one representing perfect inequality, e.g., all the income is concentrated in one individual. To give a sense of values seen in practice, according to the World Bank, a recently computed Gini index for the United States is 0.418, which sits squarely in the middle of the distribution of other countries' Gini index values.  Two countries near the extremes include India at 0.255 and South Africa at 0.630 (\url{https://data.worldbank.org/indicator/SI.POV.GINI}).


Often, the incomes used to compute the Gini index derive from data collected by a statistical organization, such as a government agency or the World Bank.  The organization may be ethically or legally obligated to protect data subjects'  confidentiality. This requirement motivates the central question of this article: how can statistical organizations release values of the Gini index that provide provable guarantees of confidentiality protection?

To release other types of statistics, many organizations are turning to differential privacy to provide such guarantees \citep{dwork:mcsherry, dwork2006differential, dwork:roth}.  For example, the Census Bureau released differentially private counts from the 2020 decennial census \citep{topdown}.  The Internal Revenue Service and Department of Education use differentially private algorithms in the College Scorecard to protect statistics related to educational and financial outcomes for students \citep{tumultlabs}. The Opportunity Insights project at Harvard University adds noise to statistics related to social mobility using an algorithm inspired by differential privacy \citep{chetty}.  To the best of our knowledge, however, researchers have not developed a differentially private algorithm for releasing the Gini index.

In this article, we develop such an  
algorithm.
To do so, we 
determine how much the Gini index can change when altering any one income in the data; this is known as the local sensitivity in the differential privacy literature \citep{dwork:roth}.  We  derive a smooth upper bound on the local sensitivity, which we use to develop the differentially private algorithm. Using theory and empirical investigations, we characterize conditions on the data and privacy demands under which the algorithm is likely to provide accurate results, and under which it is not.  We use  Bayesian inference to provide an interval estimate for the Gini index, accounting for the uncertainty introduced by the differentially private noise. 

The remainder of this article is organized as follows.  In Section \ref{sec: bg_and_def}, we review the Gini index, differential privacy, and smooth sensitivity \citep{nissim2007smooth}.  In Section \ref{sec:alg}, we present our algorithms.  In Section \ref{sec:analysis}, we provide theoretical support for these algorithms.  In Section \ref{sec:experiments}, we present empirical investigations of the properties of the differentially private Gini index.  Finally, in Section \ref{Sec: Discussion}, we conclude with ideas for future research.  


\section{Background} \label{sec: bg_and_def}
For $i=1, \dots, n$, let $z_i$ be the income (or other numerical variable) for individual $i$. Let $Z = (z_1, \dots, z_n)$.  We presume $Z$ has no missing values. 
For convenience, we define an ordered version of $Z$ as $X=(x_1, \dots, x_n)$, so that $x_1$ is the smallest value in $Z$, $x_2$ is the second smallest value in $Z$, and so on until $x_n$ is the largest value in $Z$. We refer to the position of each $x_i$ as its rank in $X$, e.g., $x_i$ has rank $i$.

\subsection{Gini Index}
    Given a set of data $Z$, the Gini index $g(Z)$ can be computed as 
\begin{equation}
    g(Z) = \frac{\sum_{i=1}^{n}\sum_{j=1}^n|z_i-z_j|}{2n^2\bar{z}}.\label{eq:giniZ}
\end{equation}
There are various ways to rewrite \eqref{eq:giniZ}. We use a representation based on the ordered data $X$. As shown by \citet{rothschild1973some}, 
\begin{equation}
g(Z) = g(X) = \frac{\sum_{i=1}^n(2i-n-1)x_i}{n\bar{x} (n-1) }. \label{eq:giniX}
\end{equation}

Several authors have explored how shifting income from one person in $Z$ to another person in $Z$ affects $g(Z)$; see \citet{gastwirth2017gini} for a summary of this work. These authors do not consider the  local sensitivity of $g(Z)$, that is,  
the effects on $g(Z)$ of adding, deleting, or changing a single observation, which are useful quantities for differential privacy. 

Several researchers have noted that top-coding, which is a commonly used disclosure protection method for income data, can result in underestimation of the Gini index \citep{armour2016using,epi, Sun03042025}. With top-coding, any income value above a selected threshold $T$ is blanked and replaced with the mean of the income values exceeding $T$, or with just $T$ itself. However,  
it is difficult to quantify how much privacy top-coding leaks about the underlying incomes, especially when other data products are released that involve income distributions, possibly without the use of top-coding. 


\subsection{Differential Privacy}

Differential privacy (DP) 
is a mathematical definition of 
privacy protection. 
Let $\mathcal{M}$  be a randomized algorithm that takes as input any database $D$ and returns an output $R$. Here, we consider databases that are vectors of $n$ real numbers like incomes.  Define neighboring databases $D$ and $D_1$ so that they differ on only one element.  To develop the DP Gini index, we work with two types of neighboring databases.  The first case is when $D$ and $D_1$ have $n$ elements in common but $D_1$ has one more or one fewer observation than $D$.  The second case is when $D_1$ is constructed by changing one element in $D$ to some other value, so that both $D$ and $D_1$ have $n$ elements. 

We say that $\mathcal{M}$ satisfies $(\varepsilon, \delta)$‑differential privacy if, for every pair of neighboring databases $(D, D_1)$ and for every measurable subset of outputs $R$, 
\begin{equation}
  \Pr\bigl(\mathcal{M}(D)\in R\bigr)\;\le\;e^{\varepsilon}\,
  \Pr\bigl(\mathcal{M}(D_1)\in R\bigr) + \delta.\label{DP}
\end{equation}
In this article, we work mainly with so-called exact or pure DP, in which $\delta=0$.

Heuristically, when $\mathcal{M}$ satisfies pure DP, for any dataset $D$, the probability that the algorithm produces any particular set of outcomes is similar regardless of whether or not any one individual's data are present in $D$. We emphasize that the definition applies over all possible instantiations of $D$, not just the observed data at hand. The degree of similarity is controlled by the parameter $\epsilon$, whereby smaller values of $\epsilon$ generally result in greater privacy protection.  However, smaller values of $\epsilon$ also generally result in randomized algorithms that introduce more noise into the outputs. Ideally, the value of $\epsilon$ is selected to achieve a satisfactory balance in accuracy and disclosure risks \citep{abowd2016economic, reiter2019differential}.  For a summary discussion of setting $\epsilon$, see \citet{kazan:reiter}.


For many DP algorithms, a key quantity is the global  
{sensitivity} of the statistic. For all  neighboring databases $D$ and $D_1$ that differ in one element, which we write as $d(D,D_1)=1$, 
and some function,
\(f: D\!\to\mathbb{R}\), we define the global sensitivity of the function  as 
\begin{equation}
      GS_{f}
      \;=\;
      \max_{D,D_1:\,d(D,D_1)=1}\;|f(D)-f(D_1)|.\label{eq:gs}
\end{equation}
For example, when $f$ is the sample mean of a set of $n$ elements each lying in $[0,1]$, and defining neighboring databases as having sample size $n$,
the 
$GS_f = 1/n$. 

When $D=(z_1, \dots, z_n)$ comprises numerical values, many DP algorithms impose upper and lower bounds to ensure $GS_f< \infty$ \citep{Kamath10022025}. That is, they presume $L \leq z_i \leq U$ for all possible values of $z_i$, clipping any values outside this range to the boundary points. In some contexts, the bounds may be derived from domain knowledge.  As examples, test scores may lie between 0 and 100; health measurements like blood pressures may be bounded by plausible values; and, monetary values in surveys may be bounded by the format of the questionnaire, e.g., the American Community Survey provides enough boxes to enter five digits for annual property taxes. 
To preserve the DP guarantee, these bounds should be selected without using the observed data at hand. When domain knowledge does not define $(L,U)$, one can allocate some (small) portion of the privacy budget to privately learn noisy bounds on the observed data. For example, one can obtain a differentially private estimate of the maximum of $D$ using the algorithm in \cite{durfee2023unbounded}.  We illustrate this approach in Section \ref{sec:alg} to set $U$ for incomes, setting $L=0$.

The $GS_f$ can be used to define randomized algorithms that satisfy DP. A commonly used algorithm is the Laplace mechanism, in which one adds to $f(D)$ a random draw sampled from a Laplace distribution 
with mean zero and scale parameter $GS_f/\epsilon$, resulting in $\tilde{g}(D)$. Unfortunately, the Laplace mechanism cannot be used directly to construct a differentially private version of $g(Z)$ with low error for reasonable privacy guarantees.  This is because the global sensitivity of the Gini index $GS_{g(Z)}=1$.  To see this, consider the 
extreme case where a single individual holds all the income. Removing that individual or changing them to have zero income reduces $g(Z)$ from $1$ to $0$.  Unless $\epsilon$ is very large, adding Laplace noise scaled to this global sensitivity  obliterates the usefulness of $\tilde{g}(Z)$.

. 

\subsection{Smooth Sensitivity}\label{sec:smooth}

The global sensitivity of the Gini index derives from extreme cases that, for all practical purposes, are fictions. For example, the Gini index for most countries lies between $0.2$ and $0.7$  \citep{charles2022gini}.  This situation, i.e., the $GS_f$ is made impractically large to accommodate unrealistic instances of neighboring databases, is not uncommon in DP.  One approach to circumvent this issue is to work with a smooth bound on the local sensitivity.  We use this strategy to construct the DP Gini index, as described  in Section \ref{sec:alg}.

For any function \(f : D\!\to\mathbb{R}\) and a fixed database $D$, the {local sensitivity} is defined as 
\begin{equation}
      {LS}_{f}(D)
      \;=\;
      \max_{D_1:\,d(D,D_1)=1}\;|f(D)-f(D_1)|. \label{eq:localsens}
\end{equation}
The local sensitivity measures how much $f(D)$ changes when removing, adding, or changing one observation in the specific $D$, leaving all other values fixed.  In contrast, the global sensitivity considers this change over all possible realizations of $D$. Thus, $GS_f$ is the supremum of $LS_f(D)$ over $D$.  Often, ${LS}_{f}(D)$ is much smaller than $GS_f$.   

It is tempting to construct a Laplace mechanism using $LS_f(D)$ in place of $GS_f$ in the scale parameter. However, in DP one presumes that the parameters of the noise distribution are public.  Knowledge of $LS_f(D)$ could introduce disclosure risks, as adversaries may be able to use $LS_f(D)$ to learn about the values in $D$.  For example, if $LS_f(D)$ is large and $f$ is the sample mean, the adversary learns that $D$ includes an outlier and, further, may be able to use $LS_f(D)$ to approximate plausible values of that outlier. As this example suggests, using a Laplace mechanism based on $LS_f(D)$ does not satisfy DP.


To avoid the risks from releasing $LS_f(D)$ yet still satisfy DP, \citet{nissim2007smooth} introduce smooth sensitivity, a framework that allows one to release $f(D)$ with instance-based additive noise, i.e., based on the  $D$ at hand.  Following their notation, we first define a smooth upper bound on $LS_f(D)$ in Definition \ref{def: smoothbound}. 

\begin{definition}[Smooth upper bound on $LS_f(D)$]  \label{def: smoothbound}
For $\beta>0$, a function $S(D)$ is a $\beta$-smooth upper bound on $LS_f(D)$ if it satisfies the following conditions.
\begin{enumerate}
\item For all $D$ of size $n$, $S(D) \geq LS_f(D)$.
\item For every pair of neighboring datasets $(D, D_1)$ such that $d(D, D_1) = 1$, $S(D) \leq e^{\beta} S(D_1)$.
\end{enumerate}
\end{definition}

Given a smooth upper bound on the local sensitivity, we follow \citet{nissim2007smooth} and define the smooth sensitivity in Definition \ref{def: smoothsens}.
\begin{definition}[$\beta$-smooth sensitivity]  \label{def: smoothsens}
For $\beta>0$ and any two databases $D$ and $D'$ of size $n$ differing in $d(D, D')$ elements, the $\beta$-smooth sensitivity of $f$ is
\begin{equation}
S_{f,\beta}(D) = \max_{D'}(LS_f(D') e^{-\beta d(D, D')}).
\end{equation}
\end{definition}



 \citet{nissim2007smooth} provide a general strategy to compute the smooth sensitivity. Let $D'$ be some other dataset of size $n$ such that $d(D, D') = k$.  Any such $D'$ has its own local sensitivity $LS_f(D')$.  For $k=0, \dots, n$, define $A^{(k)}(D)$ 
as the maximum of these local sensitivities over all $D'$, that is, 
  \begin{equation}
      A^{(k)}(D)\;=\; \max_{D' : d(D, D') \le k} {LS}_{f}(D').
\end{equation} 
    \citet{nissim2007smooth} show that the $\beta$-smooth sensitivity  from Definition \ref{def: smoothsens} is
    \begin{equation} \label{eq: smoothsensitivity}
   S_{f,\beta}(D) = \max_{k=0, 1, \dots, n} e^{-\beta k}A^{(k)}(D).
\end{equation}



Smooth sensitivity enables a variety of differentially private algorithms.  
We consider 
algorithms $\mathcal{M}(D)$ of the form, 
\begin{equation}
  \tilde{f}(D) = f(D) +  \frac{S_{f, \beta}(D)}{\alpha}\psi,
\end{equation}
where $\alpha$ is a constant and $\psi$ is a random draw from a noise distribution, both chosen to allow $(\epsilon, \delta)$-DP to hold. In particular, \citet{nissim2007smooth} show that one can achieve $(\epsilon, 0)$-DP by drawing $\psi$ from the density,  
        \begin{align}\label{eq: mechanism}
            p(\psi) \propto \frac{1}{1+|\psi|^{\gamma}}
        \end{align}
      where $\gamma > 1$ and \( \alpha=\varepsilon/(4\gamma),\,\beta=\varepsilon/\gamma\).  We use \eqref{eq: mechanism} in the DP Gini index algorithm.
        



\citet{nissim2007smooth} also show that one can use a Laplace mechanism with smooth sensitivity and satisfy $(\epsilon, \delta)$-DP.  As noted previously, we focus on pure DP, although one could use $(\epsilon, \delta)$-DP as well.  We leave comparisons of the relative accuracy and disclosure risks for the different privacy  criteria to future work.

\section{Algorithms}
\label{sec:alg}
In this section, we present and illustrate the algorithm for using smooth sensitivity to release the Gini index under pure DP.  Supporting derivations and proofs are presented in Section \ref{sec:analysis}.  Throughout, we work with the ordered version of the data $X$, presuming that we have bounds $L$ and $U$ so that $L \leq x_i \leq U$ for $i=1, \dots, n$.


The DP Gini index algorithm is comprised of three main parts. Algorithm \ref{alg:fastming} computes  the minimum value of $g(X)$ that can result after changing $k$ elements in $X$.  Algorithm \ref{alg:fastmaxg} computes the maximum value of $g(X)$ that can result after changing $k$ elements. 
Combined, the results of these two algorithms can be used to compute $A^{(k)}(X)$ for any $k$.  Algorithm \ref{alg:ssgini} integrates all the algorithms to compute a smooth upper bound for use in smooth sensitivity. 

We illustrate how the algorithms work using a toy dataset with $n=4$ individuals, $X=(x_1, x_2, x_3, x_4) = (3, 6, 7, 7.5)$, where $L=0$ and $U=10$.  Here, $g(X) \approx 0.206$. Further, we suppose that possible values of the underlying data 
are measured to the hundredths decimal place, so that each $x_i$ could have taken any of 1000 possible values.  

To find the smooth sensitivity, we need to find $A^{(k)}(X)$ for $k=1$ and $k=2$.  We note that when replacing $k=3$ or $k=4$ points, $A^{(k)}(X)=1$  
For $k=1$,  without using our algorithm, a brute force approach for finding $A^{(1)}(X)$ is to compute the local sensitivity over all possible neighboring datasets with one element changed. This requires up to $(4 \times 1000)^2$ computations, since we  construct each neighboring dataset by replacing one of the four values of $x_i$ with one of the 1000 values in $[0, 10]$, and then evaluate the local sensitivity of this dataset which requires computing the Gini index for each of its 4000 possible neighboring datasets. Following this approach for $k=2$, we would consider up to $(6 \times 1000^2)(4000)$ combinations to find $A^{(2)}(X)$, where six is the number of ways to pick two records out of four.  
As these computations suggest, this brute force approach 
becomes intractable for any reasonable sample size $n$.

With our algorithms, \textcolor{black}{for $k=1, \dots, n,$ we find the maximum and minimum possible values of the Gini index that could be attained by changing any $k$ elements in $X$.  For each $k$, we use these values to obtain a bound on $LS_{g(X')}$ for all datasets $X'$ where $d(X, X')=k$. 
With these bounds, we can compute an expression that provides $A^{(k)}(X)$ for each $k$.}

To illustrate how our algorithms greatly reduce computations, we walk through how they find the maximum and minimum values of the Gini index in our example with $n=4$ values.  We first use Algorithm \ref{alg:fastming} to find the minimum value of \(g(X')\) when \(X'\)  changes one value in \(X\), i.e., \(k=1\). For $i=1, \dots, n$, let $X_{-i} = X - x_i$ and $X_{i}'(x) = X_{-i} \cup x$ where $L \leq x \leq U$.  In Section~\ref{sec:analysis}, we prove that this minimum value for $k=1$ is achieved  when $X_{-i}$ includes only consecutive values, i.e., $(x_{1}, \dots, x_{n-1})$ or $(x_{2}, \dots, x_{n})$. Further, we show that, when $k=1$, the Gini index is minimized when setting  $x$
equal to one of the points in $X_{-i}$. Thus, we can find the minimum Gini index with an efficient search.
For our example with $n=4$, the algorithm examines all values of $x \in X_{-i}$ for 
\(X_1'(7) =(x, 6, 7, 7.5)\) and $X_4'(x) = (3, 6, 7, x)$.  After the search, we find that $X_1'(7)$ results in the smallest \(g(X_1'(7))\approx 0.054\).  The smallest values associated with the other three $X_{-i}$ include $g(X_4'(6))\approx 0.182$,  $g(X_2'(7))\approx 0.184$, and $g(X_3'(6))= 0.200$, all of which are larger than $g(X_1'(7))$.

The algorithm next finds the minimum Gini index value when replacing  \(k=2\) elements in $X$.   Let $X_{-ij} = X - \{x_i, x_j\}$ and  $X_{ij}'(x,y) = X_{-ij} \cup \{x, y\}$ where $L \leq x,y \leq U$.  As shown in Section \ref{sec:analysis}, to minimize the Gini index we should replace $k$ consecutive elements with $(x,y)$ equal to the largest value in $X_{-ij}$.  Intuitively, in the expression for the Gini index, setting $(x,y)$ to the largest value in $X_{-ij}$ makes the pairwise differences in the numerator as small as possible while making the sample mean in the denominator as large as possible (for that $X_{-ij}$). Algorithm \ref{alg:fastming} searches over all possible $\{X_{-ij}(x,y): i<n, j = i+1\}$ for the set that yields the smallest value of $g(X_{-ij}'(x,y))$.  In our example with $n=4$ observations,  
for \(k=2\), the minimum Gini index value is attained for \(X_{-12}'(7.5,7.5)=(7,7.5,7.5,7.5)\), which yields \(g(X_{12}'(7.5, 7.5))\approx 0.017\). The smallest values of the Gini index for other sequences of consecutive $i,j$ include $g(X_{23}'(7.5, 7.5)) \approx 0.176$ and $g(X_{34}'(6,6)) \approx 0.143$.

\textcolor{black}{
Algorithm \ref{alg:fastmaxg} next finds the maximum value of the Gini index for each $k$.  
For any $k$, the search for the maximum values is accomplished by replacing $k$ consecutive $x_i$ with values at one or the other extreme, i.e., with $L$ or $U$. 
In our example with \([L,U]=[0,10]\), when \(k=1\) we search for the maximum value of the Gini index over $X_{-i}'(0)$ and $X_{-i}'(10)$ for $i=1, \dots, 4.$  The search results in selecting \(X_{-3}'(0) = (3, 6, 0, 7.5)\), with \(g(X_{-3}'(0))\approx 0.515\). For \(k=2\), the search for the maximum Gini index selects $X_{-23}'(0,0) =(3, 0, 0, 7.5)$, with \(g(X_{-23}'(0,0))\approx 0.809\). For \(k=3\), the maximum value of the Gini index is attained by setting any three values to zero, which attains the maximum Gini index of 1.}


 We note that Algorithm \ref{alg:fastming} and Algorithm \ref{alg:fastmaxg} follow similar logic for larger $n$ and $k$, each searching over sets with consecutive values.
 Finally, Algorithm \ref{alg:ssgini} takes the resulting maxima and minima as inputs to an equation that provides the bounds for smooth sensitivity.

Although not essential, given the released DP Gini index $\tilde{g} = \tilde{g}(X)$ and the public DP mechanism used to construct it, analysts can use a post-processing step to obtain an interval estimate for $g(X)$ that accounts for the noise introduced by DP.  Let $g$ represent the analyst's random variable for the value of $g(X)$. We assume the analyst has some prior distribution $p(g)$ for $g(X)$.  For example, in  the empirical illustrations in Section \ref{sec:experiments}, we presume $g \sim \mathrm{Uniform}(0,1)$. Let $p(\tilde{g}|g)$ be the density of the DP Gini index conditional on confidential Gini index as determined from the specific DP mechanism, e.g., as in \eqref{eq: mechanism}.  We compute the posterior density,
\begin{align}
    p\!\left(g \mid \tilde{g} \right)&\propto p(\tilde{g}|g)\times p(g). \label{rem: bayesdp}
\end{align}
In Section \ref{sec:experiments}, we sample from \eqref{rem: bayesdp} using self-normalized importance resampling with a \(\mathrm{Uniform}(0,1)\) proposal and weights proportional to \(p(\tilde{g}\mid g)\). Other strategies could be used as well, for example, a simple grid sampler.

In addition to characterizing uncertainty about $g(X)$, the Bayesian post-processing step in \eqref{rem: bayesdp} restricts the support of inferences to the feasible region of the Gini index, $[0,1]$. Thus, if $\tilde{g}(X)$ is randomly drawn to lie outside $[0,1]$, as is theoretically possible under DP mechanisms with unbounded support, the Bayesian post-processing step 
ensures that inferences about $g(X)$ are based on feasible values. We note that this post-processing step does not use any additional privacy budget.
  

Finally, when the sample size \(n\) is large and the privacy budget \(\epsilon\) is small, the algorithms can be computationally expensive.  This may not be  problematic when releasing a single Gini index.  Nonetheless, in Section \ref{sec:analysis}, we show how to relax the smooth upper bound in Algorithm \ref{alg:ssgini} 
to reduce runtime while maintaining the DP guarantees.


\begin{algorithm}[H]\label{alg:fastming}\normalsize
\caption{\textsf{FastMinG}$(X,k)$: Find the minimum value of the Gini index when replacing $k$ elements in $X$} 
\KwIn{Sorted incomes $x_1\!\le\dots\le\!x_n$; number of elements to replace $k$ ($1\le k<n$)}
\KwOut{$\min g_{(k)}$: minimum value of Gini index when replacing $k$ elements 
}

\textbf{/* 1. prefix tables */}\\
$P_0\!\gets0,\;C_0\!\gets0,\;W_0\!\gets0$\;
\For{$i\gets1$ \KwTo $n$}{
    $w_i\gets2i-n-1$\;
    $P_i\gets P_{i-1}+x_i$;\quad  $W_i\gets W_{i-1}+w_i$;\quad   \tcp*{$P_i=\sum_{t\le i}x_t$}
    $R_i\gets R_{i-1}+i\,x_i$ \tcp*{$R_i=\sum_{t\le i}t\,x_t$}
} 
\quad $g_{\min}\gets\infty$\; 

\textbf{/* 2. slide window */}\\
\For{$i\gets0$ \KwTo $k$ }{                            
    $L \gets i+1$; \quad $R \gets n-k+i$ \tcp*{$i$ lows, $k-i$ highs}   
    $j = \text{Findminindex}(L, R)$\;
    $S_A\gets P_{j-1}-P_{L-1}$; $S_{tA}\gets (R_{j-1}-R_{L-1})-(L-1)\,S_A$ \tcp*{$S_A = \sum_{t=L}^{j-1} x_t$; $S_{tA}= \sum_{t=L}^{j-1}(t-L+1)x_t$}
    $S_B\gets P_R-P_{j-1}$; \quad $S_{tB}\gets (R_R-R_{j-1})-(j-1)\,S_B$ \tcp*{$S_B = \sum_{t=j}^{R} x_t$; $S_{tB}= \sum_{t=j}^{R}(t-j+1+a+k)\,x_t$}
    $N = [2S_{tA}-(n+1)S_A]+kx_j(2(j-L)+k-n) +[2S_{tB}+2(j-L+k)-n-1)S_N]$\;
    $D = S_A +kx_j +S_B$\;
    $g_{min} \gets \min(g_{min}, N/(nD)$
}
\Return $g_{min}$

\SetKwFunction{Find}{Findminindex}          
\SetKwProg{Fn}{function}{:}{}
\Fn{\Find{$lo,hi$} }{                
  \While{$hi-lo>3$}{$t\gets\lfloor(hi-lo)/3\rfloor$   \tcp*{Ternary search on $[lo,hi]$} $m_1\gets lo+t$; $m_2\gets hi-t$\; 
  {\textbf{If }$g(m_1)<g(m_2)$ \ \textbf{then} \ $hi\gets m_2-1$}{ \textbf{else }$lo\gets m_1+1$}}
  \Return $\arg\min_{j\in[lo,hi]} g(j)$
}
\end{algorithm}

\begin{algorithm}[H]\label{alg:fastmaxg}\normalsize
\caption{\textsf{FastMaxG}$(x,k,L,U)$: Find the maximum value of the Gini index when replacing $k$ elements in $X$} 
\KwIn{Sorted incomes $x_1\!\le\dots\le\!x_n$; bounds $L,U$; number of elements to replace $k$ ($1 \leq k < n)$}
\KwOut{$\max g_{(k)}$: maximum value of Gini index when replacing $k$ elements}

\textbf{/* 1. prefix tables */}\\
$P_0\!\gets0,\;C_0\!\gets0,\;W_0\!\gets0$\ $g_{\max}\leftarrow-\infty$\;
\For{$i\gets1$ \KwTo $n$}{
    $w_i\gets2i-n-1$\;
    $P_i\gets P_{i-1}+x_i$;\quad   \tcp*{$P_i=\sum_{t\le i}x_t$}
    $C_i\gets C_{i-1}+w_i\,x_i$ \tcp*{$C_i=\sum_{t\le i}w_t x_t$}
} 

\textbf{/* 2. slide window */}\\
\For{$s\gets 1$ \KwTo $n-k$}{               
    $S_L\leftarrow P_s$;\quad
    $S_R\leftarrow P_n-P_{s+k}$ \tcp*{$S_L = \sum_{t=1}^{s} x_t$; $S_R = \sum_{t=s+k+1}^{n} x_t$}
    \For{$i\gets 0$ \KwTo $k$ \tcp*{$i$ lows, $k\!-\ i$ highs}}{             
        $n_L\!\leftarrow\!i$;\quad
        $n_U\!\leftarrow\!k-i$\;
        $D\!\leftarrow\!S_L+S_R+n_L L+n_U U$\;
        $N\!\leftarrow\!
           L\,n_L\,(n_L-n)\;+\;
           C_s + 2n_L S_L\;+\;
           (C_n-C_{s+k}) - 2n_U S_R\;+\;
           U\,n_U\,(n-n_U)$\;
        $g_{\max}\leftarrow\max\!\bigl(g_{\max},\,N/(n\text{D})\bigr)$\;
    }
}
\Return $g_{\max}$\;
\end{algorithm}

\begin{algorithm}[H]\label{alg:ssgini}
\normalsize     
\caption{Smooth sensitivity of the Gini index under $(\varepsilon,0)$–DP}
\KwIn{Sorted incomes $x_1\!\le\dots\le\!x_n$ referred to as $X$; bounds $L,U$; number of elements to replace $k$ ($1 \leq k < n)$; privacy parameters $\varepsilon,\gamma$}
\KwOut{$S_u$ — smooth sensitivity}

$\beta \leftarrow \varepsilon / (2(\gamma+1))$\;
$g_0 \leftarrow \textsf{Gini}(X)$\tcp*{calculate Gini index $g(X)$}
$\mu_0 \leftarrow \textsf{Mean}(X)$\tcp*{calculate $\sum_{i=1}^n x_i/n$}
$ls_0 \leftarrow \textsf{LocalSensBounds}(X,L,U)$\tcp*{as shown in Equation \ref{eq: tightbound}}
$k_{\max} \leftarrow \min\!\bigl(n,\lceil-\ln(ls_0)/\beta\rceil\bigr)$\;
$S_u \leftarrow 0$\;

\For{$k \gets 0$ \KwTo $k_{\max}$ \tcp*{For each iterate, change k elements}}{
    $\min g_{(k)} \leftarrow \textsf{FastMinG}(X,k)$\ \tcp*{calculate minimal Gini }
    $\max g_{(k)} \leftarrow \textsf{FastMaxG}(X,k,L,U)$\ \tcp*{calculate maximal Gini}
    $\min \bar{x}_{(k)} \leftarrow \textsf{MinAve}(X,k,L)$ \tcp*{calculate minimal average}
    $\max \bar{x}_{(k)} \leftarrow \textsf{MaxAve}(X,k,U)$\ \tcp*{calculate maximal average}
    $ls_k \leftarrow \textsf{LocalSensBounds}\bigl(\min g_{(k)}, \max g_{(k)},
               \min \bar{x}_{(k)},\max \bar{x}_{(k)},L,U\bigr)$\; \tcp*{as shown in Equation \ref{eq: tightbound}}
    $S_u \leftarrow \max\!\bigl(S_u, e^{-\beta k}\,ls_k\bigr)$\;
}


\Return $S_u$\;
\end{algorithm}


\section{Analysis}\label{sec:analysis}



In this section, we provide theoretical support for the algorithms described in Section \ref{sec:alg}. 
In Section \ref{sec:calsmoothupperbound}, we derive smooth upper bounds on $LS_f(X)$ when changing one of its elements and on the local sensitivity of datasets differing from $X$ in $k > 1$ elements.  
These bounds rely on the minimum and maximum values of the Gini index attainable by changing up to $k$ elements in $X$, which we compute using the methods described in Section \ref{sec: minmaxkG}. In Section \ref{sec: ss_gini}, 
we formally establish the smooth sensitivity of our algorithm.
In  Section \ref{sec:speedups}, we describe two strategies to speed up computations.

\subsection{Smooth Upper Bound on Local Sensitivity}\label{sec:calsmoothupperbound}




Section \ref{subsec: suuinc} and Section \ref{subsec: suudec} respectively describe bounds on how much the Gini index can change when one element in $X$  increases or decreases. Section \ref{sec:smoothUB} presents bounds on the local sensitivity for any dataset differing from $X$ in $k > 1$ elements. 

\subsubsection{Change in Gini index when increasing one element in \texorpdfstring{$X$}{X}} \label{subsec: suuinc}


Suppose some \(x_m \in \  X\) is increased by a positive amount \(a>0\). Let $X_m' = X_{-m} \cup (x_m+a)$, where $X_{-m} = X - x_m$.  After the perturbation, the rank of $x_m+a$ among the values in $X_m'$  increases when $x_m+a \geq x_{m+1}$  and stays the same when $x_m + a < x_{m+1}$. Let the  new rank of $(x_m+a)$ in $X_{m}'$ be $m'$, where $m \leq m' \leq n$. To simplify notation, we use $X'$ in lieu of $X_{m}'$ and $g_0$ to denote $g(X)$. 

Using the definition of the Gini index in \eqref{eq:giniX}, adding $a$ to $x_m$ increases the denominator  from $(n-1)(n\bar{x})$ to $(n-1)(n\bar{x}+a)$.  The numerator  changes from 
$\sum_{i=1}^{n}(2i-n-1)x_i$ to 
\begin{equation} 
\sum_{i < m}(2i-n-1)x_i+  (2m'-n-1)(x_m+a)  + \sum_{m < i \leq m'}(2(i-1)-n-1)x_i  + \sum_{i > m'} (2i - n - 1)x_i. 
\end{equation}
Thus, for any $X'$, 
$\Delta g = g(X')-g_0$ depends on three quantities: the increment amount \( a \), the original rank \( m \), and the resulting rank \( m' \).  To emphasize this dependence, we write $\Delta g$ as $\Delta g(m, m', a)$.  Using the algebraic manipulations presented in the supplementary material, 
we can show that 
\begin{equation}
\Delta g(m, m', a) = \frac{1}{n\bar{x}+a}\left(a\left(\frac{2m'}{n-1}-w-g_0\right)+\frac{2}{n-1}\sum_{i=m}^{m'}(x_m-x_i)\right)\label{eq:deltag}
\end{equation}
where $w=(n+1)/(n-1)$. Note that \eqref{eq:deltag} can be positive or negative (or in rare cases zero).


We now consider \eqref{eq:deltag} as a function of $a$,  fixing  $m$ and $m'$.  Suppose that $a$ can take any rational value in  $[x_{m'}-x_m,x_{m'+1}-x_m)$.   When $a$ is in this interval, $\Delta g(m, m', a)$ is 
differentiable with respect to $a$.  We have    
\begin{align}
    \frac{\partial\Delta g(m,m',a)}{\partial a} 
    &= \frac{(\frac{2m'}{n-1}-w-g_0)(n\bar{x}+a)-a(\frac{2m'}{n-1}-w-g_0)-\frac{2}{n-1}\sum_{i=m}^{m'}(x_m-x_i)}{(n\bar{x}+a)^2}\notag \\
    &=\frac{1}{(n\bar{x}+a)^2}\underbrace{\left((\frac{2m'}{n-1}-w-g_0)n\bar{x}+\frac{2}{n-1}\sum_{i=m}^{m'}(x_i-x_m)\right)}_{A}. \label{eq: partialg1}
\end{align}

When the ranks \(m\) and \(m'\) are held fixed, the sign of \eqref{eq: partialg1} is the same for all values of \(a\) in the open interval. Thus, within any interval $[m, m')$, we can expect $\Delta g(m,m',a)$ to be always increasing or always decreasing.  This suggests that for any fixed $m$ and $m'$, the largest absolute value of $\Delta g(m,m',a)$ is likely to occur at the maximum or minimum possible values of $a$ that change the rank from $m$ to $m'$, i.e., when $a = x_{m'}-x_m$ or $a \approx x_{m'+1}-x_m$. 
Further, the term that determines the sign in  \eqref{eq: partialg1}, which we refer to as $A$, increases monotonically in $m'$ when $m$ is fixed.
To see this, note that $(2m'/(n-1))n\bar{x}$ is linear in $m'$ with positive slope, and  $\sum_{i=m}^{m'}(x_i-x_m)$ is nondecreasing in $m'$ because each additional term $x_i-x_m \geq 0$ when $m\leq i \leq m'$.  

We use these features of \eqref{eq: partialg1} to help determine 
 the largest possible {absolute value} of \(\Delta g(m, m', a)\) for any fixed rank \(m\). To do so, we determine the values of $m'$ that result in the largest and smallest \(\Delta g(m, m', a)\), where the smallest is of interest because \(\Delta g(m, m', a)\) could be negative.
 The result is contained 
 in Lemma \ref{lem: gchange_fixm}.
\begin{lemma}
\label{lem: gchange_fixm} 
    Fix $m\in\{1,\dots,n\}$ and let $m'$ be any integer such that $m\le m'\le n$. The following two facts hold.
\begin{enumerate}
    \item When $\tfrac{2m}{n-1} - w - g_0<0$, $\Delta g(m, m', a)$ attains a unique minimum at the integer, 
        \begin{equation}
          m'=s_{m}=
          \min\Bigl\{m'<n : 
          \bigl(\frac{2m'}{n-1}-w-g_0\bigr)n\bar x
          +\frac{2}{n-1}\sum_{i=m}^{m'}(x_i-x_m)\ge 0
          \Bigr\}.
        \end{equation}
        When $\frac{2m'}{n-1} - w - g_0 \geq 0$, $\Delta g(m, m', a)$ 
        attains its minimum at $m' = m$.
 \item $\Delta g(m,m',a)$ attains its maximum at $m'=n$. 
\end{enumerate}
\end{lemma}
\begin{proof} 
    We begin by establishing several facts about  
    the sign of \eqref{eq: partialg1} at the endpoints of the domain of $m'$. When $m'=n$,  \eqref{eq: partialg1} has a positive sign since 
    \begin{equation}
    \left(\frac{2n}{n-1}-w-g_0\right)n\bar{x}+\frac{2}{n-1}\sum_{i=m}^{n}(x_i-x_m) >0.\label{eq:lemma41a}
    \end{equation}
 When $m' = m$, the sign of \eqref{eq: partialg1} is determined by the sign of 
    $2m/(n-1)-w-g_0$.
    This expression could be negative for some values of $m$ and positive for other values of $m$, which motivates examining the two cases in the first fact of the lemma statement.  



We next show that $\Delta g(m,m',a)$ does not have discontinuities at the points defined by $a=x_{m'}-x_m$.  
Let $r$ be an integer satisfying $m < r \leq n$. As $a \to (x_{r} - x_m)^-$, we have $m'=r-1$; as $a \to (x_{r} - x_m)^+$, we have $m'=r$.  Thus, for any $r$ we have 
\begin{align} \nonumber
    &\Delta g(m,m'=r,a=(x_r-x_m)^+)-\Delta g(m,m'=r-1,a=(x_r-x_m)^-)\\
    &= \lim_{a\to {x_r-x_m}}\frac{1}{n\bar{x}+a}\left[\frac{2a}{n-1}+\frac{2(x_m-x_{r})}{n-1}\right] = 0.
\end{align}
Thus, the differences in \(\Delta g(m, m', a)\) at the endpoints of any two adjacent intervals are negligible.

Recall that $A$ from \eqref{eq: partialg1}  is nondecreasing in $m'$ when $m$ is fixed.  
We use this fact to reason about the two cases described in the lemma statement.
When $2m/(n-1)-w - g_0<0$, increasing $m$ eventually switches the sign of \eqref{eq: partialg1} from negative to positive.  
Thus, $\Delta g(m,m', a)$ decreases until some $m'$ and increases thereafter. 
The minimum occurs at the first \(m'\) for which $A$ is greater than or equal to zero, i.e., the value $s_{m}$ given in the lemma statement. 
    When $2m/(n-1)-w-g_0\geq 0$, \eqref{eq: partialg1} is always positive, so that  
    $\Delta g(m,m', a)$ increases monotonically with $m'$. 
    Hence, in this case, the minimum value of $\Delta g(m,m', a)$ occurs at \(m'=m\).  Finally, for both cases the maximum value occurs at \(m'=n\). Since $A$ is nondecreasing and positive for all $m' \geq s_{m}$, $\Delta g(m, m', a)$ continues to grow for each increment of $m'$ until reaching $n$.
\end{proof}

Lemma~\ref{lem: gchange_fixm} shows that, for each fixed \(m\), we can find the the largest value of \(|\Delta g(m,m',a)|\) by evaluating it at $m'=n$ with $a = a_{n} = U - x_{m}$ and at the relevant minimum $m'=s_{m}$ with  $a = a_{m'}=x_{m'}-x_m$.  Note the minimum value of $\Delta g(m,m,a)=0$ when $m'=m$, so we need not consider $m'=m$ for determining the largest value of \(|\Delta g(m,m',a)|\).
As a result, finding the largest value of \(|\Delta g(m,m',a)|\) over all admissible \((m,m')\) when increasing a single $x_i$ amounts to finding the larger of \(|\Delta g(m,m'=n,a_n = U-x_m)|\) and \(|\Delta g(m,m'=s_m,a_{s_m}=x_{s_m}-x_m)|\) for each \(m\), and searching among these for the largest value over $m$.


In fact, for many realizations of $X$, we need not evaluate these quantities for every value of $m=1, \dots, n$.  The maximum of \(|\Delta g(m,m',a)|\) occurs either when $(m=1, m'=s_m)$ or when $(m=m_I, m'=n)$, where $m_I = \lfloor n - (1 - g_0)(n-1)/2 \rfloor$, i.e., the greatest integer less than $n - (1 - g_0)(n-1)/2$.  
Intuitively, for $m'=s_m$, we want to  increase the smallest element in $X$ to obtain the largest decrement in the Gini index.  For $m'=n$, we want to increase the element in $x$ that allows for the largest increment in the Gini index, which generally occurs for some $1<m<n$.

To argue why this should be the case, we first fix \(m' = n\) and determine which choice of \(m\) maximizes $|\Delta g(m,m'=n,a)|$. We have 
\begin{eqnarray}
\Delta g(m,m'=n,a) &=& \frac{1}{n\bar{x}+a}\left[a(\frac{2n}{n-1} - g_0 - w)+\frac{2}{n-1}\sum_{i=m}^n(x_m-x_i)\right] \label{eq:changewithntrue}\\ 
&=& \frac{1}{n\bar{x}+a}\left[a(- g_0 + 1)+\frac{2}{n-1}\sum_{i=m}^n(x_m-x_i)\right].\label{eq:changewithn}
\end{eqnarray}
For any $m$, let $a_{m+1} = x_n - x_{m+1}$ and $a_m = x_{n} - x_m$.  When $n\bar x >> U-x_1$, which should be the case in practice as long as $n$ is reasonably large and $U$ is sensible, $\Delta g(m+1,m'=n, a_{m+1}) - \Delta g(m,m'=n, a_m)$ for $m<n$ is determined primarily by the term inside the brackets of \eqref{eq:changewithn}, i.e.,.
\begin{align}
    (x_m-x_{m+1})(1-g_0)+\frac{2(n-m)}{n-1}(x_{m+1}-x_m)) = \underbrace{\left(\frac{2(n-m)}{n-1}-1+g_0\right)}_{B}(x_{m+1}-x_m). \label{eq: changegm}
\end{align}

The sign of \eqref{eq: changegm} depends on the expression labeled $B$, which decreases strictly with $m$. This $B$ is positive when $m=1$ and eventually becomes negative. Thus, \eqref{eq:changewithn} with $m' = n$ is first monotone increasing and then monotone decreasing in $m$.
The maximum of $\Delta g(m,m'=n)$ is attained at the last integer that makes $B$ positive, which is $m = m_I$. When $m' = s_m$, the maximum of $\Delta g(m,m'=n)$ occurs when $m=1$, which changes the smallest value in $X$. A proof is in the supplementary material.

We note that computing $s_m$ and $m_I$ can be computationally expensive, since $s_m$ has no closed-form solution.  In Section \ref{sec:smoothUB}, we present a bound on the local sensitivity that avoids the need to compute these values.


\subsubsection{Change in Gini index when decreasing one element in \texorpdfstring{$X$}{X}} \label{subsec: suudec}

In this section, we consider decreasing $x_m$ by $a>0$.  The decrement could result in a value with reduced rank $1 \leq m'<m$ or a value still with rank  $m'=m$.  
We can use analyses akin to those in Section \ref{subsec: suuinc} to determine the maximum absolute change in the Gini index when reducing one $x_m$. For brevity, here we provide the upper bounds on this change and leave the details to the supplementary material. 
 
When decreasing $x_m$ by some amount $a$, the largest value of \(|\Delta g(m, m', a)|\) occurs when either \((m,m')=(n,q_m)\) or \((m,m')=(m_D,1)\). 
Here,  $q_m=  \max \bigl\{1\leq m'\leq m:(g_0 - \frac{2m'}{n - 1} + w)n\bar{x} + \frac{2}{n - 1} \sum_{i=m'}^{m-1}(x_i - x_m) \leq 0\bigr\} $ and $ m_D = \max\{m>1:(g_0+1)n\bar x+\frac{2}{n-1}\sum_{i=1}^{m-1}(x_i-x_m)\geq 0\}$. 

\subsubsection{Bounds on local sensitivity for datasets differing from  \texorpdfstring{$X$}{X} in  \texorpdfstring{$k$}{k} elements}\label{sec:smoothUB}


For smooth sensitivity, we need to compute a bound on the local sensitivity for any dataset constructed by replacing $k$ elements in $X$ with $k$ other values in $[L,U]$.  Let $X_{(kt)}$ denote one of these modified datasets, where the index $k=1, \dots, n$ indicates the number of elements replaced and the index $t$ ranges over all possible datasets with   $d(X,X_{(kt)})=k$. The elements of $X_{(kt)}$ are arranged in increasing order.  For $k=1, \dots, n$, let $\min g_{(k)}$ and $\max g_{(k)}$ denote, respectively, the minimum and maximum values of the Gini index $g(X_{(kt)})$ over the set of all feasible datasets  \{$X_{(kt)}: d(X,X_{(kt)})=k\}$ of size $n$. 
Similarly, let $\min\bar{x}_{(k)}$ and $\max \bar{x}_{(k)}$ denote, respectively, the minimum and maximum values of the average of the $n$ values in $X_{(kt)}$ over the set of all feasible datasets \{$X_{(kt)}: d(X,X_{(kt)})=k\}$ of size $n$.
To simplify our notation, we write $X_{(kt)} = (x_{t1}, \dots, x_{tn})$, $\bar{x}_t = \sum_{i=1}^n x_{ti}/n$, and $g_{t} = g(X_{(kt)})$.


As shown in Section \ref{subsec: suuinc} and Section \ref{subsec: suudec}, the maximum absolute change in $g_t$ when increasing or decreasing any single element occurs for one of four possible values of $(m, m')$.  These include $(m=m_I, m'=n)$, 
$(m=1, m'= s_m)$, $(m=m_D, m'= 1)$, and $(m=n, m'= q_m)$.  We use these facts to bound the local sensitivity for any $X_{(kt)}$.  Note that the values of $m_I$ and $m_D$ in these four candidates are computed using $X_{(kt)}$, not $X$.


For the case where $(m=m_I, m'=n)$, we have 
\begin{align}
    &|\Delta g(m=m_I,m'=n, a)| = \bigl|\frac{1}{n\bar{x}_t+a}\left(a(\frac{2n}{n-1}-g_t-\frac{n+1}{n-1})+\frac{2}{n-1}\sum_{i=m_I}^n(x_{m_I}-x_i)\right)\bigr|\notag\\
    &\le \max\!\left(
        \frac{1}{\,n\bar{x}_t+a}a\,(1-g_t),\;
        \frac{2}{(n\bar{x}_t+a)(n-1)}\sum_{i=m_I}^{n}(x_{ti} - x_{tm_I})
      \right) \notag\\    
    &< \max\!\left(\frac{(U-L)(1-\min g_{(k)})}{n\min \bar{x}_{(k)}+(U-L)},\frac{2}{(n\bar{x}_{t}+0)(n-1)}\sum_{i=1}^{n}(x_{ti}-x_{t1})
    \right)\notag\\
    &<\max\!\left(\frac{(U-L)(1 - \min g_{(k)})}{n \min\bar{x}_{(k)}+(U-L)},\frac{2(n\max \bar{x}_{(k)}-nL)}{ n\min\bar{x}_{(k)}(n-1)}\right). \label{eqa: cand1}  
\end{align}
The second step holds because of the triangle inequality. The third step  utilizes the fact that $0 \leq a \leq U$ and that the largest value of the summation term occurs when the limits go from $1$ to $n$. 
Using this upper bound for the summation allows us to circumvent searching for $m_I$---which does not have a closed-form solution---thereby avoiding a search procedure with complexity $\mathcal{O}(log(n))$. The final step holds because we maximize the numerator and minimize the denominator. 

For the case where $(m=1, m'=s_m)$ we can use almost identical logic to show that \eqref{eqa: cand1} also provides an upper bound on  $|\Delta g(m=1,m'=s_m, a)|$.

For the case where $(m=n,m'=q_m)$, we have 
\begin{align}
    &|\Delta g(m=n,m'=q_m,a)|
    = \left|\frac{1}{n\bar{x}_t-a} \left(a(g_t - \frac{2}{n-1} + 1) + \frac{2}{n-1} \sum_{i=q_m}^{n} (x_{ti} - x_{tn})\right) \right|\notag\\
    & < \max\left(\frac{(U - L)(\max g_{(k)} + 1 -\frac{2}{n-1})}{n \min \bar{x}_{(k)} - (U - L)}, \frac{2n(U - \min \bar{x}_{(k)})}{(n \min \bar{x}_{(k)} - (U - L))(n-1)}\right). \label{eqa: cand3}
\end{align}

 For the case where $(m=m_D,m'=1)$, we can show that \eqref{eqa: cand3} also
 provides an upper bound on  $|\Delta g(m=m_D,m'=1, a)|$.
 
\subsection{Efficiently Finding  \texorpdfstring{$\max g_{(k)}$}{max g(k)} and \texorpdfstring{$\min g_{(k)}$}{min g(k)} }\label{sec: minmaxkG}
As evident in \eqref{eqa: cand1} and \eqref{eqa: cand3}, our smooth upper bounds on the local sensitivity depend on minima and maxima of the Gini index.  In this section, we provide statements of theorems that offer computationally efficient ways to find these minima and maxima. 
Specifically, we show that for $k=1, \dots, n-1$, we can find $\max g_{(k)}$ by removing and replacing $n-k$ consecutively ordered elements from $X$.  We also show that, for $k= 1, \dots, n-1$, we can find $\min g_{(k)}$ by removing and replacing a total of $n-k$ consecutively ordered elements from $X$ selected from the smallest and largest values of $X$.  
Proofs of the theorems  are provided in the supplementary material.



As part of the methodology, we define $k$-maximal Gini subsets and $k$-minimal Gini subsets as follows.
\begin{definition}[$k$-maximal Gini subset]\label{def:kmvsmax}
Given a set \(X=(x_{1},\ldots,x_{n})\) of \(n\) real numbers in ascending order and an integer \(k<n\), a subset \(Q\subset X\) is called a {$k$-maximal Gini subset} of \(X\) if \(|Q| = k\) and,
  $\forall\,Q'\subset X \text{ such that } |Q'| = k \textrm{ where } Q'\neq Q,$ we have $g(Q')\le g(Q).$
\end{definition}

\begin{definition}[$k$-minimal Gini subset]\label{def:kmvsmin}
Given a set \(X=(x_{1},\ldots,x_{n})\) of \(n\) real numbers in ascending order and an integer \(k<n\), a subset \(Q\subset X\) is called a $k$-minimal Gini subset of \(X\) if \(|Q| = k\) and,
  $\forall\,Q'\subset X \text{ such that } |Q'| = k \textrm{ where } Q'\neq Q,$ we have $g(Q')\ge g(Q).$
\end{definition}

To illustrate the concepts in Definition \ref{def:kmvsmax} and Definition \ref{def:kmvsmin}, we return the the example with $X = (3, 6, 7, 7.5)$ from Section \ref{sec:alg}. The $k$-maximal Gini subset of size $k = 2$ is the pair of values $(x_i, x_j)$ whose Gini index is largest, i.e., the pair that differs most. For this $X$, that subset is $\{3, 7.5\}$.
The $k$-minimal Gini subset of size $k = 2$ is the pair of values $(x_i, x_j)$ whose Gini index is smallest, which for this $X$ is $\{7, 7.5\}$. 
In general, for any $k < n$, the $k$-maximal/minimal Gini subsets tell us how much/little inequality can exist in a sample from $X$ of size $k$.

The $k$-maximal and $k$-minimal subsets are used to prove Theorem~\ref{the: maxgconsecutive} and Theorem \ref{the: mingconsecutive}. 

\begin{theorem}\label{the: maxgconsecutive}
Given a set $X = (x_1, \dots, x_n)$ of of $n$ real numbers in ascending order, where each $x_i \in [L,U]$, and an integer $1<k<n$, the $k$-maximal Gini subset can be achieved by removing $n-k$ consecutive elements in $X$.
\end{theorem} 



\begin{theorem}\label{the: mingconsecutive}
Given a set $X = (x_1, \dots, x_n)$ of $n$ real numbers in ascending order, where each $x_i \in [L,U]$, and an integer $1<k<n$, the $k$-minimal Gini subset can be achieved by removing a total of $n-k$ elements from 
$X$ so that $Q$ comprises $k$ consecutive elements.
\end{theorem}


These theorems show that, instead of comparing all $\binom{n}{k}$ subsets with $k$ elements, we can find  $k$-maximal and $k$-minimal Gini subsets by removing  consecutive elements of the sorted list $X$. Thus, it is sufficient to analyze these structured deletions rather than every possible subset of size $k$.
In particular, the \(k\)-maximal Gini subset can be computed in \(\mathcal{O}(n-k)\) steps and the $k$-minimal Gini subset can be computed in \(\mathcal{O}(k)\) steps.  As a result, the bounds for smooth sensitivity presented in Section \ref{subsec: suuinc} and \ref{subsec: suudec} can be computed in  $\mathcal{O}(k(n-k))$ steps.



\subsection{Smooth Sensitivity of Gini Index} \label{sec: ss_gini}

Putting it all together, we can find  $A^{(k)}(X)$ for $k=1, \dots, n-1$  for use in the smooth sensitivity as defined in Section \eqref{sec:smooth}. 
\begin{theorem} \label{cor: ak_dp}
For a given $k \in \{1, \dots, n\}$, let $C_1$ be the value of the bound in \eqref{eqa: cand1} and $C_2$ be the value of the bound in \eqref{eqa: cand3}.  Then, 
\begin{align} \label{eq: tightbound}
    A^{(k)}(X) = 
\begin{cases}
\max(C_1 ,C_2) & \text{if } n\bar{x}_k - (U-L) > 0 \text{ and } \max(C_1 ,C_2) \leq1 \\
1 & \textit{otherwise}.  
\end{cases}
\end{align}
   
\end{theorem}

As noted at the end of Section \ref{sec:alg}, there may be contexts where analysts wish to compute some $A^{(k)}(X)$ with fewer computations, for example, when $n$ is very large.
Theorem \ref{thm: ss_relaxing} presents an alternative upper bound that avoids computation of $\min g_{(k)}$ and $\max g_{(k)}$.
\begin{theorem} \label{thm: ss_relaxing}
Let $IQ(X) =\frac{U-L}{\bar{x}}$.  Then, for $k=1, \dots, n$, an upper bound on $A^{(k)}(X)$ is 
\begin{align}\label{eq: relaxingbound}
A^{(k)}(X) = 
\begin{cases}
\frac{2}{n(\frac{1}{IQ(X)}-\frac{k}{n})-1} & \text{if } \frac{1}{\frac{1}{IQ(X)}-\frac{k}{n}}\leq \frac{U-L}{L} \text{ and } \frac{2}{n(\frac{1}{IQ(X)}-\frac{k}{n})-1} <1\\
1 & \textit{otherwise.}  
\end{cases}    
\end{align}
\end{theorem}
\begin{proof}
Let $X_{(kt)}$  be any ordered dataset of $n$ elements that satisfies $d(X,X_{(kt)})=k$. 
Let $C_1$ and $C_2$ be the values in \eqref{eqa: cand1} and \eqref{eqa: cand3}, respectively, computed with $X_{(kt)}$. Since $0 \leq \max g_{(k)} \leq 1$, we have 
\begin{equation}
\max(C_1, C_2) \leq \frac{2}{\frac{n}{IQ(X_{(kt)})}-1}.\label{eq:maxc1c2}
\end{equation}
It is obvious that $\min \bar{x}_{(k)} = (\sum^{n-k}_{i=1}x_i + kL)/n \geq \bar{x}-k(U-L)/n$ and $\max \bar{x}_{(k)} = (\sum^{n}_{i=k+1}x_i + kU)/n \leq \overline{x}+k(U-L)/n$. Since $L \leq \bar{x}_t  \leq U$, we have
\begin{equation}
\frac{1}{\frac{1}{IQ(X)}+\frac{k}{n}}\leq IQ(X_{(kt)}) \leq min\left(\frac{U-L}{L},\frac{1}{\frac{1}{IQ(X)}-\frac{k}{n}}\right).\label{eq:IQ}
\end{equation}
The result in \eqref{eq:IQ} implies \eqref{eq:maxc1c2}.
\end{proof}

The difference between the bounds in Theorem \ref{cor: ak_dp} and in Theorem \ref{thm: ss_relaxing} is of the order $\mathcal{O}(\frac{1-g_0}{n/IQ(X)-1})$. In practical situations, we expect this difference to be small, since the denominator typically is much larger than the numerator. It is theoretically possible, however, for the discrepancy to be large. 
To illustrate, if $n=2$ with  
$x_1 = 0.5$ and $x_2=1$, and $L=0.5$ and $U=1$, the difference between the bounds is $0.35$.  
The bounds in Theorem \ref{cor: ak_dp} imply smooth sensitivity, as stated formally in Theorem \ref{lem: accuracy_cauchy}.


\begin{theorem} \label{lem: accuracy_cauchy}
    Let $\epsilon>0$ and  $\eta \in (0,1)$. 
    Define $A^{(k)}(X)$ as in  \eqref{eq: relaxingbound} and fix $\gamma > 1$ in \eqref{eq: mechanism}. Let $\beta = \frac{\epsilon}{\gamma}$. 
    If
    \begin{equation}
    n \geq \max \big(2IQ(X)(\frac{2}{\epsilon \eta}+1),\frac{2IQ(X)\gamma}{\epsilon}ln(\frac{1}{\epsilon \eta})\big)
    \end{equation}
    then
    \begin{equation}
    S(X) = \max_{k= 1, \dots, n}e^{(-\beta k)}A^{(k)}(X) \leq \epsilon \eta.
    \end{equation}
\end{theorem}
\begin{proof}
    First, suppose $k \leq \frac{n}{2IQ(X)}$. Since $n \geq 2IQ(X)(\frac{2}{\epsilon \eta}+1)$, we have 
    \begin{equation}
    \exp(-\beta k)A^{(k)}(X) \leq \max_{k= 1,...,n}A^{(k)}(X) =\max_{k= 1,...,n}\frac{2}{n(\frac{1}{IQ(X)}-\frac{k}{n})-1} \leq \frac{2}{\frac{n}{2IQ(X)}-1} \leq \epsilon \eta.
    \end{equation}
    Now suppose $k > \frac{n}{2IQ(X)}$. Then, $\exp(-\beta k ) \leq \exp(-\frac{\beta n}{2 IQ(X)})=\exp(-\frac{\epsilon n}{2\gamma IQ(X)})$. Since $n \geq \frac{2IQ(X)\gamma}{\epsilon}ln(\frac{1}{\epsilon \eta})$ and $A^{(k)}(X) \leq 1$ for all $X$ and $k$, we have
    \begin{equation}
    \max_{k= 1, \dots, n}\exp(-\beta k)A^{(k)}(X) \leq \max_{k= 1, \dots, n} \exp(-\beta k) \leq \exp{-\frac{\epsilon n}{2 IQ(X)\gamma}\leq \epsilon \eta}.
    \end{equation}
\end{proof}

\subsection{Computational Speedups}\label{sec:speedups}

We use two computational strategies to make the algorithms run more efficiently, namely (i) pruning the search over values of \(k\) and (ii) maintaining auxiliary in-memory data structures.  The latter supports dynamic updates to the Gini index after modifying \(k\) records, thereby avoiding repeated computations using all of $X$.

With regard to {pruning the search over \(k\),} we first note that \(A^{(k)}(X)\le 1\).  Thus, 
        $e^{-\beta k}\,A^{(k)}(X)\;\le\;e^{-\beta k}$.
Let  $k^{\star}
        \;=\;
        \min\left(\,k\in\mathbb N \,\big|\, e^{-\beta k}\le A^{(0)}(X)\right)$.
    For every \(k \ge k^{\star}\) we  have
    \(e^{-\beta k}A^{(k)}(X)\le A^{(0)}(X)\), so no \(k>k^*\) can lead to a different maximization in \eqref{eq: smoothsensitivity}. Consequently, we can restrict the search range to \(0 \le k \le k^{\star}\).

 With regard to the dynamic Gini maintenance, a naive implementation of Algorithm~\ref{alg:fastming} and Algorithm~\ref{alg:fastmaxg} could slide a length-\(k\) window across $(x_1, \dots, x_n)$ and re-sort the records after advancing each window. This would result in a cost of \(\mathcal{O}\!\bigl((n-k)\,k\log n\bigr)\).  
  However, this is not necessary.  Only the $k$ elements in the window can change, while the remaining \(n-k\) items keep their original ranks. We  leverage this fact in the search strategy  
 in  Algorithm \ref{alg:fastming}.  
  We also save time by not recomputing certain statistics that stay fixed.  In particular, we  maintain the order statistics $P_i = \sum_{t<i}x_t$, $C_i = \sum_{t<i}(2t-n-1)x_t$, and $R_i = \sum_{t<i}tx_t$ defined in Algorithm \ref{alg:fastmaxg} and Algorithm \ref{alg:fastming}.   Finally, we compute the Gini index efficiently using these stored statistics; see the supplementary material for the expression.

\section{Experiments}\label{sec:experiments}
In this section, we illustrate and evaluate the algorithm for generating a DP version of the Gini index using an open-source dataset comprising genuine incomes (Section \ref{subsec: realdataset}) and  hypothetical data with fixed values of $g(X)$ (Section \ref{subsec: sytheticdataset}). 
Because we evaluate the repeated sampling properties of the algorithm, we use the 
bounds from Theorem \ref{thm: ss_relaxing}  to facilitate efficient computation.



\subsection{Illustrations Using Genuine Incomes} \label{subsec: realdataset}
 
 As the data $Z$, we use income values from the 2024 Current Population Survey public use file \citep{cps_asec_2024_puf_csv}. This reports calendar-year 2023 income (named ``PTOTAL'' on the file). We 
 restrict $Z$ to individuals aged 16 and older, resulting in $n=$ 115,777 individuals. 
 The resulting $g(Z) = 0.574.$ Note that we do not consider the survey weights when computing $g(Z)$; see Section \ref{Sec: Discussion} for further discussion of survey weights.
 

\begin{figure}[t]
    \centering
    \includegraphics[width=0.8\linewidth]{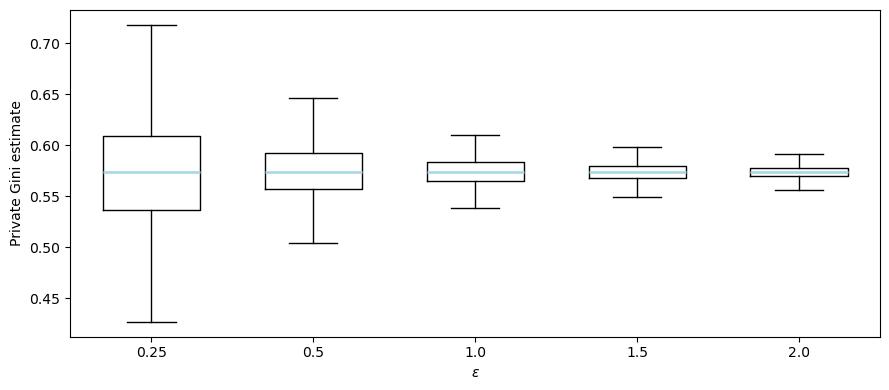}
    \caption{Distribution of the DP Gini estimator under various privacy budgets using the Current Population Survey data as $Z$. Box plots summarize 1,000,000 independent draws of $\tilde{g}(Z)$.}
    \label{fig:realdata}
\end{figure}

\begin{figure}
    \centering
    \includegraphics[width=0.6\linewidth]{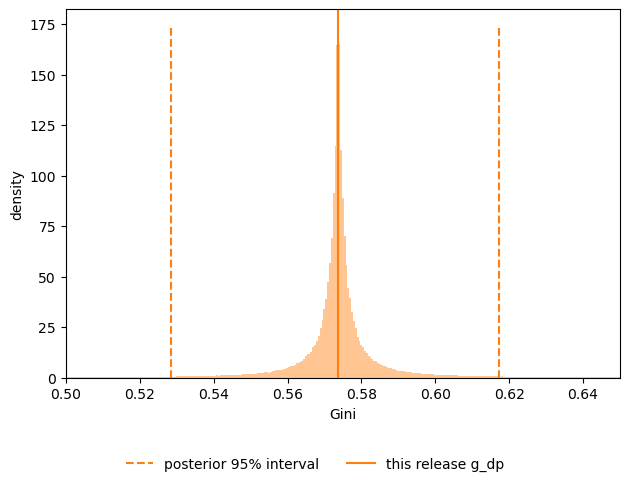}
    \caption{Posterior distribution $p(g| \tilde{g})$ when $\tilde{g}(Z)=0.574$.  Here, $\tilde{g}(Z)$ derives from one application of the DP algorithm where $Z$ is the Current Population Survey data and
     $\epsilon = 2$.  Dashed vertical lines indicate central 95\% posterior credible intervals. }
    \label{fig:glvssm}
\end{figure}

We apply the DP Gini index algorithm with privacy budgets $\epsilon \in \{0.25, 0.5, 1, 1.5, 2\}$. We set the lower bound on incomes $L=0$.  Since we do not know an upper bound  for these incomes, we use an additional $0.15$ 
of privacy budget, i.e., under composition increase the total $\epsilon$ by 0.15, to estimate a value that exceeds the maximum in $Z$, $x_n =$ 2,108,379,  with high probability.  To do so, we modify the AboveThreshold algorithm in \citet{durfee2023unbounded}. As originally designed, this algorithm iteratively guesses at the maximum value in a dataset.  To begin the algorithm, we generate a differentially private version $\tilde n$ of $n$, which we do via the Laplace mechanism with privacy budget $\epsilon_1 = 0.075$. The algorithm then starts with some initial guess at the maximum and counts the number of observations in $X$ that exceed that guess, adding a modest amount of noise to the count to ensure differential privacy.  We add this noise via the Laplace mechanism with privacy budget $\epsilon_2 = 0.075$. The algorithm stops at the first instance the noisy count exceeds $\tilde n$.
The value of the guess at this instance is the differentially private  estimate of the maximum value in $X$, $\tilde{x}_n$.  

The AboveThreshold algorithm can produce an $\tilde{x}_n < x_n$. Setting $U$ equal to such an $\tilde{x}_n$ would require clipping any $x_i>\tilde{x}_n$ at $\tilde{x}_n$, which would introduce bias in the differentially private Gini index. To reduce the risk of using a $U < x_n$, we use a post-processing step and set $U = 2.5 \tilde{x}_n$.   This represents a conservative inflation, in that  $\tilde{x}_n$ generated by our implementation of the AboveThreshold algorithm with $(\epsilon_1, \epsilon_2)$ is highly unlikely to underestimate $x_n$ by 60\%, per the results in  \citet{durfee2023unbounded}. 
 Across all of our simulation runs, 
 $2.5 \tilde{x}_n$ always exceeds $x_n$ with a maximum of approximately 
 3,640,000.
 We run this modified AboveThreshold algorithm to set $U$ in  each simulation run.

For each of the five values of $\epsilon$, we generate  1,000,000 values of $\tilde{g}(Z)$ using the algorithms in Section \ref{sec:alg}. As the DP mechansim, we use \eqref{eq: mechanism} with $\gamma = 2$. As evident in Figure \ref{fig:realdata},  
the DP Gini index is approximately unbiased, as  the averages of the draws of $\tilde{g}(Z)$ are centered at $g(Z)$. The unbiasedness results because
the DP mechanism does not clip the observed data (as long as $U>x_n$) and uses mean-zero noise. 
As  $\epsilon$ increases, the variances of $\tilde{g}(Z)$ shrink. Even at $\epsilon = 0.25$, which is a strong privacy guarantee, we are likely to generate a $\tilde{g}(Z)$ within just a few points of $g(Z).$ 

Figure ~\ref{fig:glvssm} illustrates the posterior distribution $p(g \mid \tilde{g})$ for a single value of $\tilde{g}= 0.574$ when $\epsilon = 2$. The posterior distribution is computed using the Bayesian post-processing procedure in \eqref{rem: bayesdp}. As evident in the figure, the posterior distribution provides a sense of the uncertainty about $g(Z)$.


\subsection{Illustrations Using Simulated Data}
  \label{subsec: sytheticdataset}

In this section, we evaluate the performance of the algorithm from Section \ref{sec:alg} on hypothetical datasets with pre-specified values of $g(X)\in\{0.2, 0.5, 0.7\}$, sample sizes $n \in \{10000, 100000, 1000000\}$, and privacy budgets $\epsilon \in\{0.5,1,2\}$.  We first describe how we generate the hypothetical datasets, followed by the results.

\subsubsection{Method for generating simulated data}

The bounds in Theorem \ref{thm: ss_relaxing} are based on $IQ(X)$, which requires values for $L$, $U$, and $\bar{x}$. Thus, when using these bounds in  simulations, 
we need not generate individual records in $X$ that yield a $g(X) \in \{0.2, 0.5, 0.7\}$. 
Instead, we find the minimum value of $IQ(X)$ that is consistent with a particular $g(X)$, as shown in Proposition \ref{prop: gini-range}.
The proof of Proposition  \ref{prop: gini-range} is in the supplementary material.

\begin{proposition}\label{prop: gini-range}
Let $X$ have $n\ge2$ non-negative elements with mean $\bar x =n^{-1}\sum_{i=1}^n x_i>0$.  For any $g\in[0,1)$,
\begin{equation}\label{boundcases}
\inf\{IQ(X): g(X)=g\,\}\;=\;
\begin{cases}
4g, & 0\le g\le \tfrac12,\\[4pt]
\dfrac{1}{1-g}, & \tfrac12< g<1.
\end{cases}
\end{equation}
\end{proposition}


Using Proposition~\ref{prop: gini-range}, we find the minimum $IQ(X)$ that corresponds to  each prescribed $g(X)$ and use that $IQ(X)$ in Theorem \ref{thm: ss_relaxing}.
We also examine the performance of the DP algorithm when $IQ(X)$ exceeds the minimum in \eqref{boundcases}. We perform the simulations setting  $g(X) = 0.5$ and $n=100000$. 

\subsubsection{Results}

Figure~\ref{fig: diffgini} summarizes the absolute differences $|\tilde{g}(X) - g(X)|$  across 1,000,000  draws of  $\tilde{g}(X)$ when $IQ(X)$ is at its minimum value. 
The absolute differences are typically small, even when $\epsilon = 0.5$. The absolute differences decrease as both $n$ and $\epsilon$ increase.  The absolute differences tend to be largest when $g(X) = 0.7$.  This arises mainly because the minimum $IQ(X)$ bound increases with $g(X)$, which results in greater noise variance. In particular, as $g(X)\!\to\!1$, 
we would expect $X$ to include a few very large values (i.e., near $U$) of $x_i$. Even small perturbations of these large values can inflate the local sensitivity.  These results suggest that the DP Gini index may require relatively large $\epsilon$ or $n$ to give useful estimates in cases where $g(X)$ is near one.  The results also suggest that, when implementing the DP Gini index algorithms,  setting $U$ and $L$ to make $IQ(X)$ as small as possible---which can be done in practice by allocating some privacy budget to estimate a noisy bound on $x_n$---can help reduce the magnitudes of the absolute differences.



\begin{figure}[t]
    \centering
    \includegraphics[width=1\linewidth]{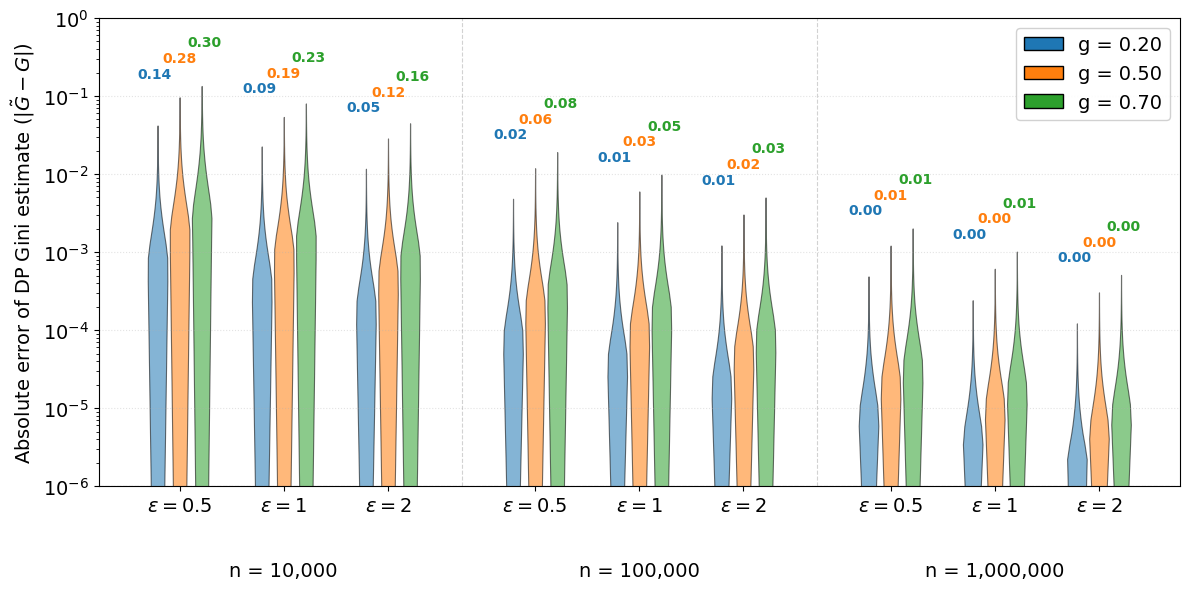}
    \caption{Simulated absolute errors in the differentially private Gini index $\tilde g(X)$ using hypothetical datasets. Each plot summarizes the absolute errors of 1,000,000 independent replications. Plots summarize up to the 99\% percentiles of the absolute errors.  The largest errors among the remaining 1\% are displayed on top of the plots.
    The vertical axis is scaled as $\log\!\big(\max\{|\tilde g(X)-g(X)|,10^{-6}\}\big)$ to distinguish  small errors.} 
    \label{fig: diffgini}
\end{figure}

Figure \ref{fig:diff_range} examines sensitivity to values of $IQ(X)$ that range from $2$ to $100$.  Note that $\inf\{IQ(X): g(X)=0.5\,\}=2$ when $n = 100000$. The root mean squared error (RMSE) increases steadily with $IQ(X)$, illustrating the potential downsides of using a large normalized range in the privacy mechanism, particularly when $\epsilon$ is small.
We also see that 
increases in $\epsilon$ yield sizable reductions in RMSE across all $IQ(X)$. Thus, 
for $X$ where the normalized range is expected to be high, it may be beneficial to increase the privacy budget, essentially trading privacy protection for analytic utility.  
Regardless, when $g(X)$ is close to 1, the normalized range is large, making the RMSE potentially large as well. We discuss a potential remedy for this issue in Section \ref{Sec: Discussion}.


\begin{figure}[t]
    \centering
    \includegraphics[width=0.8\linewidth]{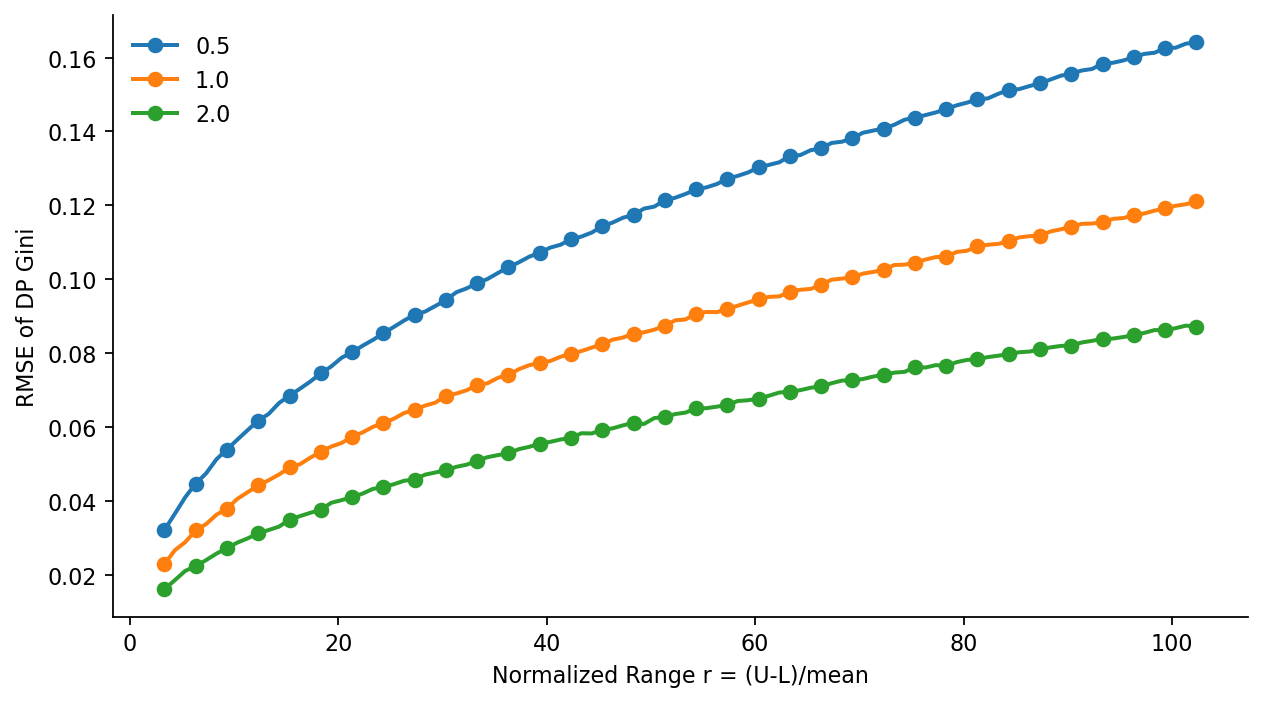}
    \caption{RMSE for different values of the normalized range, $r=U-L/\bar{x},$ varying from $\inf\{IQ(X):g(X)=0.5\}$ to $100$.  Here, we fix $g(X) = 0.5$ and $n = 100000$.}
    \label{fig:diff_range}
\end{figure}

\section{Discussion} \label{Sec: Discussion}

The differentially private Gini index algorithm can perform well when the normalized range $IQ(X)$ of $X$ is not too large.
When this is not the case, we need a relatively large sample size or  privacy budget to release reliable results. With small sample sizes, extreme observations can result in large local sensitivity, potentially pushing it close to worst-case levels. In such settings, we suspect that no unbiased differentially private mechanisms can deliver simultaneously strong privacy and high utility. 

One motivation for using smooth sensitivity is to avoid clipping some $x_i$ at a value below $x_n$, as clipping induces bias. However, for $X$ with high $IQ(X)$, a practical trade-off between privacy and utility may be necessary. One possible strategy is to allocate part of the privacy budget to privately estimate some upper quantile $q$ of the values in $X$, say the 99th percentile, release its differentially private estimate $\tilde{q}$, and compute the Gini index using values of $X$ clipped at quantile $q$ (i.e., set all $x_i>\tilde{q} = \tilde{q}$) using the DP algorithm presented here. This procedure could substantially reduce the DP noise in the released Gini index in cases with large $IQ(X)$. 

Analysts may be able to use a post-processing method to adjust for the effects of clipping. 
The general idea is for analysts to presume some model for the tail values beyond $\tilde{q}$, 
for example, a Pareto distribution or a shape-constrained Lorenz continuation, anchored at $(q,\tilde{q})$ and constrained to reproduce the released Gini statistic. The collection of tails that satisfy these constraints determines an identified interval for the Gini index of the data without clipping. Similar approaches have been suggested to adjust disclosure-protected Gini index values for the effects of topcoding \citep{LarrimoreBurkhauserFengZayatz2008,BurkhauserFengJenkinsLarrimore2008}.
We leave investigation of this suggestion to future work.

Our algorithms do not incorporate survey weights or other kinds of weights.  The literature on differential privacy does not yet provide effective methods for handling survey weights, particularly when they are adjusted for nonresponse or calibration.  This is because the sensitivity can be greatly increased by using survey-weighted analyses \citep{reiter2019differential, drechsler}. It may possible to modify our algorithms to include other kinds of observation-level weights, particularly those that are fixed features of each record in the population. In fact, downweighting extreme values could reduce the local sensitivity of the Gini index. Developing methods for computing a differentially private, weighted Gini index is a direction for future research.

\bigskip
\begin{center}
{\large\bf SUPPLEMENTARY MATERIAL}

\end{center}

\begin{description}
\item[Supporting Theory:]  Proofs and derivations supporting the results (pdf file) \label{appendix: proofs}
\item[Python-package for DP Gini Algorithm:] Python-package DP-Gini containing code to implement the algorithms. The package also contains all datasets used in Section \ref{sec:experiments}. Code is also available at 
\href{https://github.com/YancyLan/DP_Gini}{\url{https://github.com/YancyLan/DP_Gini}}
\end{description}


\bibliographystyle{agsm}
\bibliography{Bibliography-MM-MC}

\appendix
\input{appendix_1}

\end{document}

%% file: appendix_1.tex
\section{Appendix}

This section includes supplementary material for the main text.  Table \ref{tab:notation} includes a list of all notation.  Section \ref{sec:proof4.1} includes proofs and derivations relevant for Section 4.1 in the main text. Section \ref{sec:proof4.2} includes proofs and derivations relevant for Section 4.2 in the main text. Section \ref{sec:proof4.4} includes expressions relevant for Section 4.4 in the main text. 
Section \ref{sec:proofsec5} includes proofs and derivations relevant for Section 5 in the main text.

\begin{table}[htbp]
  \centering
  \caption{Notation}
  \label{tab:notation}
  \begin{tabularx}{\linewidth}{C Y}
    \toprule
    \textbf{Symbol} & \textbf{Meaning} \\
    \midrule
    z_i & income for individual $i$ \\
   Z = (z_1, \dots, z_n) & incomes for $n$ individuals\\
   x_i & $i$th largest value  in $Z$\\
    X=(x_1,\dots,x_n) & ordered version of $Z$ \\
    L,U & The lower and upper bounds of $X$\\ 
    w & $(n+1)/(n-1)$\\
    g(X) & Gini index for the dataset $X$ (we sometimes write $g_0$) \\
    LS_{f}(D) & local sensitivity for function $f$ and database $D$ \\
    A^{(k)}(X) & The maximum of the local sensitivity over all $X'$, where $d(X,X')=k$.\\
    X_{-m} & dataset $X$ with the $m$th entry removed \\
   a & the amount of increment/decrement in $x_m$ when changing it\\
    X'_m & $X_{-m}\cup (x_m+a)$ when increasing $x_m$, \\ & $X_{-m}\cup (x_m-a)$ when decreasing $x_m$\\
    m' & The new rank of $x_m +a$ or $x_m-a$ in $X'_m$ \\
    \Delta g(m, m', a) & $g(X_m')-g(X)$\\
    X_{(kt)}={(x_{t1}, \dots, x_{tn})} & Dataset constructed by replacing $k$ elements in $X$ with $k$ other values in $[L,U]$\\
    \overline{x}_t & the average of  $X_{(kt)}$ \\
    g_t & the Gini index of $X_{(kt)}$\\
    \bottomrule
  \end{tabularx}
\end{table}

\subsection{Proofs from Section 4.1 in Main Text}\label{sec:proof4.1}

\subsubsection{Derivation of expression for \texorpdfstring{$\Delta g(m,m',a)$}{g(m,m',a)}}
Recall that $g(X) = \sum_{i=1}^n(2i-n-1)x_i/\left((n-1)n\bar{x}\right)$. Let $w=(n+1)/(n-1)$. In Section 4.1, we show in (13) that by changing some $x_m$ to $x_m+a$ where $a>0$, we have  
\begin{equation}
\Delta g(m, m', a) = \frac{1}{n\bar{x}+a}\left[a\left(\frac{2m'}{n-1}-w-g_0\right)+\frac{2}{n-1}\sum_{i=m}^{m'}(x_m-x_i)\right].\label{eq:deltagappend}
\end{equation}
When adding $a$ to $x_m$,  
the  denominator $(n-1)n\bar{x}$ in $g(X)$ changes to $(n-1)(n\bar{x}+a)$.  The numerator $\sum_{i=1}^{n}(2i-n-1)x_i$ in $g(X)$ changes to 
\begin{equation} 
\sum_{i < m}(2i-n-1)x_i+  (2m'-n-1)(x_m+a)  + \sum_{m < i \leq m'}(2(i-1)-n-1)x_i  + \sum_{i > m'} (2i - n - 1)x_i. 
\end{equation}
Thus, we have
\begin{align}
g(X_m') &=
\frac{1}{(n\bar{x}+a)(n-1)}
\Bigg(\sum_{i < m}(2i-n-1)x_i+  (2m'-n-1)(x_m+a)  \notag \\
& \quad + \sum_{m < i \leq m'}(2(i-1)-n-1)x_i  + \sum_{i > m'} (2i - n - 1)x_i\Bigg) \notag\\
&=\frac{1}{(n\bar{x}+a)(n-1)}
\Bigg(
  \sum_{i=1}^{n}(2i-n-1)x_i
  -  (2m-n-1)x_m   \notag\\
&\quad
  + (2m'-n-1)(x_m+a)
  - \sum_{m<i\leq m'}2x_i
\Bigg) \notag\\
&= \frac{n\bar{x}(n-1)g_0+(2m'-n-1)(x_m+a)-(2m-n-1)x_m-\sum_{m<i\leq m'}2x_i}{(n\bar{x}+a)(n-1)}\notag\\
& = \frac{n\bar{x}(n-1)g_0+(2m'-n-1)a+2(m'-m)x_m-\sum_{m < i \leq m'}2x_i}{(n\bar{x}+a)(n-1)} \notag\\
&=g_0+\frac{-(n-1)g_0a+(2m'-n-1)a+2\sum_{m < i \leq m'}(x_m-x_i)}{(n\bar{x}+a)(n-1)} \notag\\
&=g_0+\frac{1}{n\bar{x}+a}\left(a\left(\frac{2m'}{n-1}-\frac{n+1}{n-1}-g_0\right) + \frac{2}{n-1}\sum_{i=m}^{m'}(x_m-x_i) \right). \label{eq: xmchange}
\end{align}

\subsubsection{When \texorpdfstring{$m'=s_m$}{m'=sm} the maximum of  \texorpdfstring{$\Delta g$}{g} occurs when \texorpdfstring{$m=1$}{m=1} } \label{sec: proofsm}
In Section 4.1.1 of the main text, we claim that $\Delta g(m,s_m,a = x_{s_m}-x_m)$ reaches a maximum when $m=1.$  In this section, we provide additional details on this claim. As a reminder, we define
 \begin{equation}
          s_{m}=
          \min\Bigl\{m'<n : 
          \bigl(\frac{2m'}{n-1}-w-g_0\bigr)n\bar x
          +\frac{2}{n-1}\sum_{i=m}^{m'}(x_i-x_m)\ge 0
          \Bigr\}. \label{eq:changewithsm}
        \end{equation}
Increasing $m$ by one, we also define 
 \begin{equation}
          s_{m+1}=
          \min\Bigl\{m'<n : 
          \bigl(\frac{2m'}{n-1}-w-g_0\bigr)n\bar x
          +\frac{2}{n-1}\sum_{i=m+1}^{m'}(x_i-x_{m+1})\ge 0
          \Bigr\}. \label{eq:changewithsm1}
        \end{equation}
When \(m' = s_m\),
we have 
\begin{eqnarray}
\Delta g(m,m'=s_m,a) &=& \frac{1}{n\bar{x}+a}\left[a(\frac{2s_m}{n-1} - g_0 - w)+\frac{2}{n-1}\sum_{i=m}^{s_m}(x_m-x_i)\right] \label{eq:changewithntrueappend}
\end{eqnarray} with $a = x_{s_m}-x_m$


From the definition of $s_m$, we define 
\begin{align}
    c_m = \left( \frac{2s_m}{n-1} - w - g_0 \right) n\bar{x} + \frac{2}{n-1} \sum_{i=m}^{s_m} (x_i - x_m)  \geq 0,
\end{align}
where $c_m$ is non-negative per the definition of $s_m$. Using $c_m$, $\Delta g(m,s_m,a)$ can be written as 
\begin{align} \label{eq: simplifyg}
    \Delta g(m, s_m, a) =\underbrace{(\frac{2s_m}{n-1} - w- g_0)} _A - \frac{c_m}{n\bar{x} + a}.
\end{align}
We now argue that $s_{m+1}\ge s_m$.
 For any $m' \ge m+1$, let
 \begin{align}
&F(m,m') = \bigl(\frac{2m'}{n-1}-w-g_0\bigr)n\bar x+\frac{2}{n-1}\sum_{i=m}^{m'}(x_i-x_m), \\
&F(m+1,m') = \bigl(\frac{2m'}{n-1}-w-g_0\bigr)n\bar x+\frac{2}{n-1}\sum_{i=m+1}^{m'}(x_i-x_{m+1}).
\end{align}
Notice that $\sum_{i=m+1}^{m'} (x_i - x_{m+1}) - \sum_{i=m}^{m'} (x_i - x_m) = -(m'-m)(x_{m+1} - x_m) \leq 0$. Thus, $F(m+1,m') \leq F(m,m')$. This inequality implies that if $F(m+1,m')\geq0$, we must have $F(m,m')\geq 0$. Since $s_{m+1}$ is the smallest $m'$ for which $F(m+1,m')\geq0$, we have $s_{m+1}\geq s_m$. 

Since $c_m$ is small relative to $(n\bar x+a)$, 
$\Delta g(m,s_m,a)$ mainly depends on the term labeled $A$ in \eqref{eq: simplifyg}. Since $s_{m+1} \geq s_m$, we have $s_m \geq s_1$ and hence, by comparing the terms labeled $A$ for $m$ and $m=1$, $\Delta g(m,s_m,a=x_{s_m}-x_m) \geq \Delta g(1,s_1,a=x_{s_1}-x_1)$. When $m=1$, $\Delta g(1,s_1,a=x_{s_1}-x_1)$ is negative. The definition of $s_m$ also ensures $\Delta g(m, s_m, a)$ to be negative. 
We conclude that when $m'=s_m$ the maximum of $|\Delta g(m, s_m, a)|$ occurs when $m=1$.

\subsubsection{Proofs for Section 4.1.2}
In Section 4.1.2 of the main text,  we claim that when we decrease $x_m$ by some amount $a$, the largest value of \(|\Delta g(m, m', a)|\) occurs when either \((m,m')=(n,1)\) or \((m,m')=(m_D,1)\). In this section, we provide support for this claim.

Suppose the observation \(x_m \in X\) is decreased by a positive amount \(a>0\). After the perturbation, the rank of $x_m-a$ among the values in $X_m'$ decreases to $m'$ where $1\leq m' \leq m$. Denote the new dataset as $X'_m$. 
Using algebraic manipulations akin to those used to compute \eqref{eq:deltagappend}, we can show that
\begin{align}
    g(X'_m) &= \frac{1}{(n\bar{x} - a)(n - 1)}(\sum_{i < m'} (2i-n-1)x_i + (2m' - n - 1)(x_m - a) \notag \\
     &\quad + \sum_{m' \leq i < m} (2(i+1)-n-1)x_i + \sum_{i>m} (2i - n - 1)x_i) \\ 
    &= \frac{n\bar{x}(n-1)g_0 - (2m - n - 1)x_m + (2m' - n - 1)(x_m - a) + \sum_{m' \leq i < m} 2x_i}{(n\bar{x} - a)(n - 1)} \\
    & = g_0+ \frac{ a(n - 1)g_0 + \sum_{m' \leq i < m} 2x_i - (2m - n - 1)x_m + (2m' - n - 1)(x_m - a)}{(n\bar{x} - a)(n - 1)}\\
    &= g_0 +\frac{ a\left((n - 1)g_0- 2m' + n + 1\right) + 2\sum_{i=m'}^{m} (x_i - x_m)}{(n\bar{x} - a)(n - 1)}.
\end{align}
Thus, dividing through by $(n-1)$, we have 
\begin{align}
    \Delta g(m,m',a) &= 
    \frac{a(g_0 - \frac{2m'}{n - 1} + w) + \frac{2}{n-1}\sum_{i=m'}^{m} (x_i -x_m)}{n\bar{x} - a}.
\end{align}
Within any open interval $(m', m]$, $\Delta g(m, m', a)$ is continuous in $a$ and differentiable. Considering \(m\) and \(m'\) as constants, we have 
\begin{align}
    \frac{\partial g(m,m',a)}{\partial a} 
    &= \frac{(g_0 - \frac{2m'}{n - 1} + w)(n\bar{x} - a) + a(g_0 - \frac{2m'}{n - 1} + w) + \frac{2}{n - 1}\sum_{i=m'}^{m}(x_i - x_m)}{(n\bar{x} - a)^2}  \\
    &= \frac{1}{(n\bar{x} - a)^2} \left( \left(g_0 - \frac{2m'}{n - 1} + w\right)n\bar{x} + \frac{2}{n - 1} \sum_{i=m'}^{m}(x_i - x_m) \right).\label{eq: gmmaD}
\end{align}

We now provide a lemma that is the analogue of Lemma 4.1 from the main text.  Here, we call it Lemma \ref{lem: ming}.

\begin{lemma} \label{lem: ming}
    Fix  $m\in\{1,\dots,n\}$ and let $m'$ be any integer such that $1\leq m'\leq m$. The following two facts hold.
    \begin{enumerate}
        \item When $g_0-\frac{2m}{n-1}+w <0$, $\Delta g(m,m',a)$ attains a unique minimum at the integer
        \begin{equation}
         m'=q_m=  \max \bigl\{1\leq m'\leq m:(g_0 - \frac{2m'}{n - 1} + w)n\bar{x} + \frac{2}{n - 1} \sum_{i=m'}^{m}(x_i - x_m) \leq 0\bigr\}.
       \end{equation}
       When $g_0-\frac{2m}{n-1}+w \geq 0$, $\Delta g$ attains its minimum at $m' = 1$.
       \item $\Delta g(m, m', a)$ attains its maximum at $m'=m$.
    \end{enumerate}
\end{lemma}
\begin{proof}
When $m'=m$,  $\frac{\partial g(m,m',a)}{\partial a} = \frac{1}{(n\bar{x} - a)^2}[(g_0-\frac{2m}{n-1}+w)n\bar{x}]$.
When $g_0-\frac{2m}{n-1}+w>0$, increasing $m'$ eventually switches the sign of $\frac{\partial g(m,m',a)}{\partial a}$ from negative to positive. Otherwise, when $g_0-\frac{2m}{n-1}+w\leq 0$, the sign of 
$\frac{\partial g(m,m',a)}{\partial a}$ 
remains positive for all $m'<m$.
When $m'=1$, we have $\frac{\partial g(m,m',a)}{\partial a} = \frac{1}{(n\bar{x} - a)^2}[(g_0+1)n\bar{x}+\frac{2}{n-1}\sum_{i=1}^{k-1}(x_i-x_k)]$. 
\end{proof}

For a fixed $m$, the largest change of the Gini index occurs at $m' = 1$ or at $m' = q_m$.
We first consider $m'=1$. From the sign of the partial derivative 
in \eqref{eq: gmmaD}, we can show that the largest change happens when $m_D = \max\{m>1:(g_0+1)n\bar x+\frac{2}{n-1}\sum_{i=1}^{m-1}(x_i-x_m)\geq 0\}$. 

We next consider $m'= q_m$.
Using the definition of $q_m$, define 
\begin{align}
    d_m = \left( g_0 - \frac{2m'}{n-1} + w \right) n\bar{x} + \frac{2}{n-1} \sum_{i=m'}^m (x_i - x_m)  \leq 0,
\end{align}
We can rewrite $\Delta g( m,q_m, a)$: 
\begin{align} \label{eq: simplifygd}
    \Delta g(m, q_m, a) = \underbrace{- \left( g_0 - \frac{2q_m}{n-1} + w \right) }_{A}+ \frac{d_m}{n\bar{x} - a}.
\end{align}
Since $d_m$ is negative and relatively small, the value of $\Delta g(m, q_m, a)$ mainly depends on the term labeled $A$ in \eqref{eq: simplifygd}. We now argue that $q_{m+1}\ge q_m$. We define $F(m, m') = \left( g_0 - \frac{2m'}{n-1} + w \right) n\bar{x} + \frac{2}{n-1} \sum_{i=m'}^m (x_i - x_m).$ Notice that $F(m+1, m')-F(m, m') = \frac{2}{n-1} (m - m' + 1) (x_m - x_{m+1}) \leq 0$. Therefore, $q_{m+1}\geq q_m$.  We also have $\Delta g(n,q_n,a =x_n-x_{q_n}) \leq \Delta g(m,q_m,a=x_m-x_{q_m})$ for any $m < n$. Since $\Delta g(m, q_m, a) < 0$ under the definition for $q_m$, we conclude our claim.


\subsection{Proofs for Section 4.2}\label{sec:proof4.2} 


In this section we prove Theorem \ref{the: maxgconsecutiveappend} and Theorem \ref{the: mingconsecutiveappend} from the main text.  To do so, we first prove a lemma that characterizes the effects of inserting an element on the Gini index.  Let $X_{+x}=(x_1,\dots ,x_{m-1}, x, x_{m}, \dots, x_{n})$ be the ordered dataset after inserting  \(x\in[x_{m-1},x_m)\) at position $m$. 

\begin{showas}{4}{2}{lemma}{definition} \label{lem: addg} 
 Fix \(m\in\{1,\dots ,n+1\}\). 
Let
$\Delta g(X_{+x}) = g(X_{+x})-g_0$.  There exists a unique integer \(m_A\in[1,n+1]\) 
such that
\begin{equation}
\frac{\partial\Delta g(X_{+x})}{\partial x}<0\;\textrm{ when } m<m_A,\qquad
\frac{\partial\Delta g(X_{+x})}{\partial x}>0\; \textrm{ when } m > m_A.
\end{equation}
\end{showas}

\begin{proof}
Recognizing that $-(n+1)-1=-n-2$, we can write $g(X_{+x})$ as
\begin{eqnarray} 
g(X_{+x}) &=& \frac{1}{n(n\bar x+x)}\left(\sum_{1 < i \leq m-1}(2i - n - 2)x_i + (2m - n - 2)x + \sum_{m \leq i \leq n}(2(i+1) - n - 2)x_i\right) \notag\\
&=& \frac{1}{n(n\bar x+x)}\left(\sum_{i=1}^n(2i - n - 1)x_i + (2m - n-2)x - \sum_{1 \leq i \leq m-1}x_i +  \sum_{m \leq i \leq n}x_i \right) \notag\\
&=& \frac{1}{n(n\bar x+x)}\left(g_0(n\bar{x}(n-1)) + (2m - n-2)x - \sum_{1 \leq i \leq m-1}x_i +  \sum_{m \leq i \leq n}x_i\right). \label{eq:gx+}
\end{eqnarray}
Hence, subtracting $g_0$ from \eqref{eq:gx+}, we have 
 \begin{equation}
    \Delta g(X_{+x})    = \frac{1}{n(n\bar{x}+x)}\left(-g_0n(\bar{x}+x) + (2m - n-2)x - \sum_{1 \leq i \leq m-1}x_i + \sum_{m \leq i \leq n}x_i\right).\label{eq:deltag2}
    \end{equation}


    Viewed as a function of $x$, the rank $m$ of the added element is a step function that increases each time $x$ exceeds some $x_i \in X$.
Consider the rank $m$ of $x$ as fixed, and consider values of $x$ such that $(x_{m-1} < x \leq x_m)$. We take the derivative of 
    $\Delta g(X_{+x})$ with respect to $x$.
\begin{align}
    &\frac{\partial \Delta g(X_{+x})}{\partial x} = \frac{1}{(n\bar{x} + x)^2}  \underbrace{\left[( \frac{2m - 2}{n}-1)n\bar{x} - \sum_{i = m}^{n}\frac{x_{i}}{n}  + \sum_{1 \leq i \leq m-1}\frac{x_i}{n}\right] }_{\text{(I)}} \label{eq: addg} 
\end{align}
With \(m\) held constant, the sign of 
\(\partial \Delta g(X_{+x})/\partial x\) is the same for all \(x\): it is either negative or positive. The expression labeled (I) in \eqref{eq: addg} controls the sign of the derivative. This (I) term increases monotonically in $m$.

When $x$ has rank $m=n+1$ in $X_{+x}$, we have
    \begin{equation}
    \frac{\partial \Delta g(X_{+x})}{\partial x} =
    \frac{1}{(n\bar{x} + x)^2}
    \left[n\bar{x} + \sum_{1\leq i <n}\frac{x_i}{n}\right] > 0.\label{eq:addmax}
    \end{equation}
 When $x$ has rank $m=1$ in $X_{+x}$, we have
    \begin{equation}
    \frac{\partial \Delta g(X_{+x})}{\partial x} =
    \frac{1}{(n\bar{x} + x)^2}
    \left[-n\bar{x}-\sum_{i=1}^n\frac{x_i}{n}\right] < 0. \label{eq:addmin}
    \end{equation}
Thus, there exists a point 
$x_{m_A}$ such that $\frac{\partial \Delta g(X_{+x})}{\partial x} > 0$ for all $x > x_{m_A}$, and $\frac{\partial \Delta g(X_{+x})}{\partial x} < 0$ for all $x < x_{m_A}$.  

We next show that $\Delta g(X_{+x})$ does not have discontinuities at the points defined by $x_i \in X$.  Let $r$ be an integer satisfying $1 \le r < n$.  As $x \to x_r^-$, we have $m = r$; as $x \to x_r^+$, we have $m = r+1$. We write  $X_{x_r^-}$ as the version of $X_{+x}$ including  $x=x_r^-$ and $X_{x_r^+}$ as the version of $X_{+x}$ including  $x=x_r^+$.  Thus, for any $r$ we have 
\begin{align*}
   &\Delta g\!\left(X_{x_r^+}\right)
    - \Delta g\!\left(X_{x_r^-}\right) 
    = \lim_{x \to x_r} 
    \frac{1}{n\bar{x} + x_r}
    \left[\frac{2x_r}{n} - \frac{2x_r}{n}\right] = 0.
\end{align*}

Thus, the differences in \(\Delta g(X_{+x})\) at the endpoints of any two adjacent intervals are negligible.  This completes the proof of Lemma \ref{lem: addg}.
\end{proof}

Note that though we write the proof for Lemma \ref{lem: addg} 
using $X$ which comprises $n$ records, the lemma holds for any dataset $Q$ comprising $k>1$ records. We use this fact to help prove Theorem 4.1 and Theorem 4.2 from the main text.  As a reminder, we can write the Gini index for any dataset $Q$ comprising $k$ elements $(y_1, \dots, y_k)$ as
\begin{equation}
    g(Q) = \frac{\sum_{i=1}^{k}\sum_{j=1}^k|y_i-y_j|}{2n^2\bar{y}}.\label{eq:giniZappend}
\end{equation}

\begin{showas}{4}{1}{theorem}{theorem}
\label{the: maxgconsecutiveappend}
Given a set $X = (x_1, \dots, x_n)$ of $n$ real numbers in ascending order, where each $x_i \in [L,U]$, and an integer $1<k<n$, the $k$-maximal Gini subset can be achieved by removing $n-k$ consecutive elements in $X$.
\end{showas}

\begin{proof}




First, consider $k=2$.  Clearly, the $2$-maximal subset $Q$ includes $x_1$ and $x_n$, the smallest and largest elements of $X$.  We obtain this set by removing the $n-2$ consecutive values from $X$ excluding $x_1$ and $x_n$.



Next consider $k=3$.  The $3$-maximal subset $Q$ includes $x_1$ and $x_n$, since $|x_1 - x_n|$ is the largest possible contribution we can add to the numerator of \eqref{eq:giniZappend}.  When adding a third element, we want to make the numerator as large as possible to make $g(Q)$ as large as possible.  This is done by including either $x_2$ or $x_{n-1}$.  Thus, we obtain $Q$ by removing $n-3$ consecutive ``interior'' elements from $X$.


To prove the case for $k>3$, 
we use a proof by contradiction. Let the subset $Q=(y_1, \dots, y_k)$  of $X$ have the largest Gini index among all subsets of $X$ of size $k$.  Assume that $Q$ contains one element $y_m$, where $m$ is the rank among the $k$ elements in $Q$, that is not among the $k+1$ largest or $k+1$ smallest values in $X$, in which case the lemma is false.   We will show that we can replace  $y_m$ with some $x_i \in X$ to increase the Gini index.  
This contradicts the assumption that $Q$ is a $k$-maximal Gini subset; hence, we have a proof by contradiction. All the elements of $Q$ must be among the $k$ largest or $k$ smallest values in $X$.


Under our posited assumption about the composition of $Q$, we can write it as $(y_1, \dots, y_k)= (x_1, \dots, x_l, y_m, x_u, \dots, x_n)$.  Here,  $x_l$ and $x_u$ are the lower and upper endpoints of each run of consecutive elements.  We use $l$ and $u$ to refer to their indices. 

We now consider $Q_{-m} = Q - y_m$. Lemma \ref{lem: addg} ensures that, if we set add a point $x$ with rank $m$ to $Q_{-m}$, there is some $m_A$ such that $g(Q_{+x})$ decreases when $m_A>m$ and $g(Q_{+x})$ increases when $m_A<m$. 


First, consider the case when $m_A>m$.  Then, adding any element $x$ to $Q_{-m}$ selected from $\{x_{l+1}, \dots, x_{u-1}\}$ decreases $g(Q_{+x})$.  We therefore want to add $x$ that results in the smallest decrement of $g(Q_{+x})$. Because of the monotonic decrease in \eqref{eq: addg} with $x$, we select $x_{u-1}$.
Second, consider  the case when $m_A < m$.  Then, adding any element $x$ to $Q_{-m}$ selected from $\{x_{l+1}, \dots, x_{u-1}\}$ increases $g(Q_{+x})$.  We therefore want to add $x$ that results in the largest increment of $g(Q_{+x})$. Because of the monotonic increase in \eqref{eq: addg} with $x$, we therefore should select $x_{l+1}$.  When $m=m_A$, we can select one of  $x_{l+1}$ or $x_{u-1}$ in lieu of $x_m$ and find a larger  Gini index.

In any of these cases, we can find a $Q_{+x}$ comprising $k$ elements such that $g(Q_{+x}) \geq g(Q)$.  Thus, we have a contradiction: a $Q$ that includes some $y_m$ not among the $k$ smallest or $k$ largest values in $X$ cannot be the $k$-maximal subset.

For showing the contradiction, we need not consider any other $Q'$ including more than one element $y_m$ that is not among the $k$ smallest or $k$ largest elements. This is because, for any such $Q'$, we have $g(Q) \geq g(Q')$, since Lemma \ref{lem: addg} indicates that we can increase the Gini index by iteratively replacing elements in the ``interior'' with elements among the $k-1$ smallest or largest elements in $X$.
\end{proof}


Using Theorem \ref{the: maxgconsecutiveappend}, we can search over all possible $n-k$ sequences of consecutive elements in $X$ to find the $k$-maximal Gini subset.
We now prove Theorem \ref{the: mingconsecutiveappend}.
\begin{showas}{4}{2}{theorem}{theorem}
\label{the: mingconsecutiveappend}
Given a set $X = (x_1, \dots, x_n)$ of $n$ real numbers in ascending order, where each $x_i \in [L,U]$, and an integer $1<k<n$, the $k$-minimal Gini subset can be achieved by removing a total of $n-k$ consecutive elements from 
$X$ so that $Q$ comprises $k$ consecutive elements.
\end{showas}


\begin{proof}
   

 First consider $k=2$. Clearly, the $2$-minimal subset $Q$ must include two consecutive elements in $X$, e.g., $Q = \{x_i, x_{i+1}\}$.  For any other $Q'$ that does not have two consecutive elements of $X$, e.g., $Q'= \{x_{i-1}, x_i\}$, we can decrease the Gini index by replacing one of the elements, e.g., replace $x_{i-1}$ with $x_i$. 
 


Now consider $k \geq 3$ elements from $X$, i.e., we  remove $n-k$ elements from $X$.  
We use a proof by contradiction. Assume that a   subset $Q=(y_1, \dots, y_k)$  of $X$ has the smallest Gini index among all subsets of $X$ of size $k$.  Assume that $Q$ contains one element $x_m$, where $m$ is the rank among the $n$ elements in $X$, that is not consecutive with any of the other $k-1$ elements in $Q$.   We will show that we can replace  $x_m$ with some $x_i \in X$ that is consecutive with the remaining elements in $Q$ to decrease the Gini index.  This contradicts the assumption that $Q$ is a $k$-minimal Gini subset; hence, we have a proof by contradiction. All the elements of $Q$ must be consecutively ordered in $X$.

Under our posited assumption about the composition of $Q$, we consider two possibilities for $Q$,  namely $(y_1, \dots, y_k)= (x_m, x_l, \dots, x_u)$ or $(y_1, \dots, y_k)= (x_l, \dots, x_u, x_{m})$. Here,  $x_l$ and $x_u$ are the lower and upper endpoints of the set of consecutive elements in $Q$.  We use $l$ and $u$ to refer to their indices. For notational convenience, when $m<(l-1)$, we refer to $(x_m, x_l, \dots, x_u)$ as $Q_l$.  When $m>(u+1)$ we refer to $(x_l, \dots, x_u, x_{m})$ as $Q_{u}$.  Let $Q_{-m}$  be the $k-1$ element vector obtained by removing $x_m$ from either $Q_l$ or $Q_u$. 

First, consider adding some $x_i \in X$ such that $i<l$ to $Q_{-m}$, resulting in $Q'$.  To compute the possible changes in the Gini index $\Delta g(Q')$, we can apply the expression in \eqref{eq:deltag2} with $m=1$. It is clear that $\Delta g(Q')$ decreases the most when we make $x$ in that expression as large as possible while maintaining $m=1$; that is, we add $x_{l-1}$. 
 Thus, when adding one element to $Q_{-m}$, the set of $k$ consecutive integers $(x_{l-1}, \dots, x_u)$ has the smallest Gini index.  However, this contradicts our assumption that $Q_l$ has the smallest Gini index.  Therefore, the assumption is untrue.

Likewise, we can adding some $x_i \in X$ such that $i>u$ to $Q_{-m}$, which we again call  $Q'$.  To compute the possible changes in the Gini index $\Delta g(Q')$, we can apply the expression in \eqref{eq:deltag2} with $m=n+1$.  We see that $\Delta g(Q')$ increases the least when we make $x$ in that expression as small as possible while maintaining $m=n+1$; that is, we add $x_{u+1}$. Thus, when adding one largest element to $Q_{-m}$, the set of $k$ consecutive integers $(x_{l}, \dots, x_u, x_{u+1})$ has the smallest Gini index.  However, this contradicts our assumption that $Q_u$ has the smallest Gini index.  Therefore, the assumption is untrue.



For showing the contradiction, we need not consider any other $Q$ that includes more than one element $y_m$ that is not consecutive with the remaining elements. This is because for any such $Q$, 
we can decrease the Gini index by iteratively replacing the non-consecutive elements with elements that are consecutive to smallest or largest elements in $X$.
\end{proof}
Using Theorem \ref{the: mingconsecutiveappend}, we can search over all possible $k$ sequences of consecutive elements in $X$ to find the $k$-minimal Gini subset.

\subsection{Expression for Gini Index from Section 4.4}\label{sec:proof4.4}

In our algorithms, we compute the Gini index after perturbing $k$ elements.  To compute the maximum Gini index, after we remove $n-k$ consecutive elements with the first element starting at rank $s$, with $j$ elements set equal to  $L$ and $k-j$ elements set equal to  $U$, 
  the Gini index of the new dataset $X_{(kt)}$ equals
     \begin{align}
     g(X_{(kt)}) = \frac{Lj(j-n)+C_s+2j P_s
      +(C_n-C_{s+k})
     -2(k-j)(P_n-P_{s+k})
     +U(k-j)(n-k+j)}{P_s+P_n-P_{s+k}+jL+(k-j)U}.\end{align}
To compute the minimum Gini index, after we remove $i$ elements from the lower tail and $k - i$ elements from the upper tail, with the remaining middle block running from rank $L = i + 1$ to $R = n - k + i$, and then set all $k$ removed elements equal to the value $x_j$ for some $j$ with $L \leq j \leq R$, the Gini index of the new dataset $X_{(kt)}$ equals
\begin{align}
g(X_{(kt)}) &= \frac{1}{P_R - P_i + k\, x_j} \Big(2\left[(R_{j-1} - R_i) - i(P_{j-1} - P_i)\right] - (n+1)(P_{j-1} - P_i) \notag\\
  &\quad + k\, x_j(2(j - i - 1) + k - n) 
  + 2\left[(R_R - R_{j-1}) - (j-1)(P_R - P_{j-1})\right] \notag\\
  &\quad + (2(j - i - 1 + k) - n - 1)(P_R - P_{j-1})
\Big).
\end{align}
    These facts allow us to compute the Gini index with \(\mathcal{O}(1)\) time costs. 

\subsection{Proof for Section 5}\label{sec:proofsec5}

We use Proposition 5.1 in the main text to design several simulation studies. This proposition allows us to examine the effects on the differentially private Gini index of increasing the $IQ(X)$, and thereby the effects of increasing $U-L$, given a fixed value of the Gini index.  We now provide a proof of that proposition.
\begin{showas}{5}{1}{proposition}{proposition}
Let $X$ have $n\ge2$ non-negative elements with mean $\bar x = n^{-1}\sum_{i=1}^n x_i>0$. For any $g\in[0,1)$,
\begin{equation}\label{eq: boundcases}
\inf\{IQ(X): g(X)=g\,\}\;=\;
\begin{cases}
4g, & 0\le g\le \tfrac12,\\[4pt]
\dfrac{1}{1-g}, & \tfrac12< g<1.
\end{cases}
\end{equation}
\end{showas}
\begin{proof}
To begin, we note that $IQ(X) = (U - L) / \bar{x}$ and $g(X)$
are both scale invariant.  That is, suppose we create $Y=(y_1, \dots, y_n) = (cx_1, \dots, cx_n)$ for some constant $c>0$, and we let $(L_y, U_y)=(cL, cU)$.  Then, $IQ(Y)=IQ(X)$ and $g(Y)$ = $g(X)$. 
Hence, for this proposition, we presume the data have $\bar{x} = 1$. 
With $\bar{x}=1,$ $g(X) = \sum_{i,j} |x_i - x_j|/n^2$. 

Consider $X$ in which some fraction $\alpha$ of its elements equal $L$ and the remaining fraction $(1- \alpha)$ of its elements equal $U$. We can obtain any $g(X)$ by specifying $(L,U)$ accordingly. 
In particular, we have 
\begin{equation}
g(X) = \alpha(1 - \alpha)(U - L).\label{eq:IQandG}
\end{equation}
For fixed $g(X)=g$, we can minimize $IQ(X) = (U-L)$ by maximizing $\alpha(1 - \alpha)$. Since $\alpha(1 - \alpha) \leq 1/4$, we have  $IQ(X) \geq 4g$. However, this solution only applies when $0 \leq g \leq 1/2$. 
When $g > 1/2$, setting $\alpha = 1/2$ would force $L = 1 - 2g < 0$, which is infeasible since by assumption the elements of $X$ are non-negative. 
Setting $L = 0$, the constraint on $\bar x$ implies that $U = 1/(1 - \alpha)$. Hence, in these cases, $IQ(X) = U = 1/(1 - \alpha)$. Using \eqref{eq:IQandG}, we find  
$IQ(X) = 1/(1 - g)$. 
\end{proof}